\documentclass[oneside]{amsart}
\pdfoutput=1

\usepackage{geometry}
\usepackage[T1]{fontenc}
\usepackage[utf8]{inputenc}
\usepackage{lmodern}
\usepackage[dvipsnames]{xcolor}
\usepackage{tikz,tikz-3dplot}
\usetikzlibrary{positioning,arrows.meta,decorations.pathreplacing,decorations.pathmorphing,decorations.markings}

\newcommand{\CC}{{\mathbb C}}

\newcommand{\FF}{{\mathbb F}}

\newcommand{\QQ}{{\mathbb Q}}
\newcommand{\RR}{{\mathbb R}}

\newcommand{\ZZ}{{\mathbb Z}}
\newcommand{\ee}{\mathrm{e}}
\newcommand{\ii}{\mathrm{i}}

\newcommand{\dd}{\mathrm{d}}

\newcommand{\mm}{{\overline{m}}}

\newcommand{\zz}{{\overline{z}}}

\newcommand{\sV}{\mathcal{V}}

\newcommand{\VGint}{{V_G^\mathrm{int}}}
\newcommand{\VGext}{{V_G^\mathrm{ext}}}

\renewcommand{\Im}{\mathop\mathrm{Im}}
\renewcommand{\Re}{\mathop\mathrm{Re}}

\newcommand{\mybox}{\hspace{-2ex}\qed}
\theoremstyle{plain}
\newtheorem{thm}{Theorem}
\newtheorem{lem}[thm]{Lemma}

\newtheorem{defn}[thm]{Definition}
\newtheorem{ex}[thm]{Example}

\title{Graphical functions with spin}
\date{}
\author{Oliver Schnetz}
\address{Oliver Schnetz\\
II. Institut f\"ur Theoretische Physik\\
Luruper Chaussee 149\\
22761 Hamburg, Germany}
\email{schnetz@mi.uni-erlangen.de}

\tikzset{
	ad/.style={line width=1pt},
	fun/.style={line width=1pt, postaction={decorate},
		decoration={markings,mark=at position .55 with {\arrow[scale=.5,>=triangle 45]{>}}}},
}

\begin{document}

\begin{abstract}
The theory of graphical functions is generalized from scalar theories to theories with spin, leading to a numerator structure in Feynman integrals. The main part of this
article treats the case of positive integer spin, which is obtained from spin $1/2$ theories after the evaluation of $\gamma$ traces.

As an application (in this article used mainly to prove consistency and efficiency of the method), we calculate primitive Feynman integrals in Yukawa-$\phi^4$ (Gross-Neveu-Yukawa)
theory up to loop order eight.
\end{abstract}
\maketitle

\section{Introduction}
Originally, the theory of graphical functions was developed to analyze the number theory of primitive scalar Feynman integrals (Feynman periods or Feynman residues)
in four-dimensional $\phi^4$ theory \cite{BK,Census,gf}.
The results, using the Maple implementation {\tt HyperlogProcedures} \cite{Shlog}, led to the discovery of a connection between quantum
field theory (QFT) and motivic Galois theory, the coaction conjectures \cite{coaction,Bcoact1,Bcoact2}.

Later, in order to handle full QFTs, the theory of graphical functions was extended to non-integer dimensions \cite{numfunct}. The result of this
extension was the calculation of the $\phi^4$ beta function up to loop order seven in the minimal subtraction scheme \cite{7loops}. The field anomalous dimension was calculated
up to eight loops.

In 2021, a collaboration with Michael Borinsky led to the extension of graphical functions to even dimensions $\geq4$ \cite{gfe}. Application to six-dimensional
$\phi^3$ theory showed that the number content of $\phi^3$ theory is similar (or equal) to the number content of $\phi^4$ theory. This supports
the optimistic hope that the geometry behind the number content of $\phi^4$ theory is universal for all renormalizable QFTs. In Section \ref{sectYukawa} of this article we will see
that an extension of $\phi^4$ theory to fermions via a three-point Yukawa interaction does not seem to enlarge the number content.
It is important to note that there exist strong indications from a tool called the $c_2$-invariant that the number
content of Feynman integrals from non-renormalizable theories with vertex degrees $\geq5$ is vastly more generic than that of $\phi^4$ theory \cite{SchnetzFq,BSmod,Sc2}.
So, from a mathematical point of view, the conjectured sparsity of QFT numbers is quite
mysterious.

In addition, complete calculations in six-dimensional $\phi^3$ theory became possible and led to
record-breaking six-loop results for the beta and gamma functions \cite{5lphi3,Radcortalk,6lphi3}.

After these breakthroughs, it seemed desirable to generalize the theory of graphical functions to theories with positive spin.
Such theories have a numerator structure in the Feynman integrands, which significantly increases the complexity. The main tool in this context is integration by parts (IBP) with the Laporta algorithm (see, e.g., \cite{IBP} and the references therein). The IBP method is very powerful, but scales very badly with the loop order. In fact, in $\phi^4$ theory,
IBP is not helpful, but in $\phi^3$ theory, IBP can be used effectively. In a suitable setup, at high loop orders, the theory of graphical function (which is inherently IBP-free) is a valuable addition to the pool of QFT calculation methods.

In this article, we present the fundamental theory of graphical functions with positive spin. An early version of this article is \cite{ngft}. This article contains corrections, proofs, more results and many more examples. In particular, the new Section \ref{sectYukawa} on Yukawa-$\phi^4$ theory is the first
application of graphical functions outside spin zero.

Algorithms and results will be included in upcoming versions of {\tt HyperlogProcedures} \cite{Shlog}.

\section*{Acknowlegements}
The author is supported by the DFG-grant SCHN 1240/3-1. He thanks Sven-Olaf Moch
for discussions and encouragement. The author also thanks Simon Theil who contributed to an
early version of this work.
\section{Propagators}
Let the dimension
\begin{equation}\label{Dl}
D=2\lambda+2=2n+4-\epsilon>2.
\end{equation}
We do not assume that $D$ is an integer because analytic continuation in the dimension is the most frequently used method to regularize Feynman integrals that are divergent
in integer dimensions. Mathematical proofs in non-integer dimensions are typically hard \cite{BM1,BM2,BM3,BM4}. A reliable workaround is to use integration formulae
in integer dimensions and treat the dimension as a parameter. We follow this strategy and do not prove identities in non-integer dimensions.
Notice, however, that the core identity for graphical functions, the effective Laplace equation (\ref{eqF1fromf}) in the scalar case (\ref{DeltaD}) has been proved
in non-integer dimensions \cite{parG}.

In most calculations the dimension $D$ has to be even for $\epsilon=0$ ($n=0,1,2,\ldots$) because a theory of graphical functions in odd dimensions has not yet been developed.

We define the spin $k$ propagator $Q^\alpha_\nu(x,y)=Q^{\alpha_1,\ldots,\alpha_k}_\nu(x,y)$ in numerator form by

\begin{equation}\label{numform}
\begin{tikzpicture}[scale=1.4]

\begin{scope}[local bounding box=prop1,decoration={
	markings,
	mark=at position 0.6 with {\arrow{stealth}}}]

	\coordinate (XX) at (-1,0);
	\coordinate (YY) at (0,0);
	\coordinate (MM) at (-.5,0);

	\draw[line width=1pt,postaction=decorate](XX) -- (YY);

	\draw[fill = black] (XX) circle (1.5pt);
	\draw[fill = black] (YY) circle (1.5pt);

	\node [above=.1 of MM] {$\alpha_1,\dots,\alpha_k$};
    \node [left=.1 of XX,scale=1.1] {$Q_\nu^\alpha(x,y)=$};	
	\node [below=.1 of MM,scale=.9] {$\nu$};
	\node [below=.1 of XX,scale=1.1] {\textit{x}};
	\node [below=.1 of YY,scale=1.1] {\textit{y}};
\end{scope}

	
\begin{scope}[xshift=70,local bounding box=prop2,decoration={
	markings,
	mark=at position 0.55 with {\arrow{stealth}}}]
	\coordinate (XX) at (-1,0);
	\coordinate (YY) at (0,0);
	\coordinate (MM) at (-.5,0);
	
	\draw[line width=1pt,postaction=decorate](YY) -- (XX);
	\draw[fill = black] (XX) circle (1.5pt);
	\draw[fill = black] (YY) circle (1.5pt);

	\node [above=.1 of MM] {$\alpha_1,\dots,\alpha_k$};
	\node [below=.1 of MM,scale=.9] {$\nu$};
	\node [left=.1 of XX,scale=1.1] {$(-1)^k$};	
	\node [below=.1 of XX,scale=1.1] {\textit{x}};
	\node [below=.1 of YY,scale=1.1] {\textit{y}};
\end{scope}

\begin{scope}[xshift=130,local bounding box=eq]
	\coordinate (XX) at (0,0);
	\node [above=-0.5 of XX,scale=1] { ${\displaystyle\frac{(y^{\alpha_1}-x^{\alpha_1})\cdots(y^{\alpha_k}-x^{\alpha_k})}{\|y-x\|^{2\lambda\nu+k}}}\ .$};
\end{scope}
	
\path(prop1.east) -- (prop1-|prop2.west)  node[midway]{=};
\path(prop2.east) -- (prop2-|eq.west)  node[midway]{=};
\end{tikzpicture}
\end{equation}
We use the multi-index notation $\alpha=\alpha_1,\ldots,\alpha_k$ with $k=|\alpha|$ throughout the article.

Note that we adopt a convention where the orientation is given by end point minus initial point. This is customary in mathematics, while in QFT one more often uses initial point minus end point. Both conventions differ by a minus sign for odd $k$.

We also define a differential form of the propagator $Q_{\nu;\alpha}(x,y)=Q_{\nu;\alpha_1,\ldots,\alpha_k}(x,y)$ where the indices are subscripts,

\begin{equation}
\begin{tikzpicture}[scale=1.4]
\begin{scope}[local bounding box=prop1,decoration={
	markings,
	mark=at position 0.56 with {\arrow{stealth}}}]
	\coordinate (XX) at (-1,0);
	\coordinate (YY) at (0,0);
	\coordinate (MM) at (-.5,0);
	
	\draw[line width =1pt,postaction=decorate](XX) -- (YY);
	\draw[fill = black] (XX) circle (1.5pt);
	\draw[fill = black] (YY) circle (1.5pt);
	
	\node [below=.1 of MM,scale=.9] {$\nu;\alpha$};	
    \node [left=.1 of XX,scale=1.1] {$Q_{\nu;\alpha}(x,y)=$};	
	\node [below=.1 of XX,scale=1.1] {\textit{x}};
	\node [below=.1 of YY,scale=1.1] {\textit{y}};
\end{scope}

\begin{scope}[xshift=65,local bounding box=prop2,decoration={
	markings,
	mark=at position 0.58 with {\arrow{stealth}}}]
	\coordinate (XX) at (-1,0);
	\coordinate (YY) at (0,0);
	\coordinate (MM) at (-.5,0);
	
	\draw[line width=1pt,postaction=decorate](YY) -- (XX);
	\draw[fill = black] (XX) circle (1.5pt);
	\draw[fill = black] (YY) circle (1.5pt);
	
	\node [below=.1 of MM,scale=.9] {$\nu;\alpha$};	
	\node [left=.1 of XX,scale=1.1] {$(-1)^k$};	
	\node [below=.1 of XX,scale=1.1] {\textit{x}};
	\node [below=.1 of YY,scale=1.1] {\textit{y}};
\end{scope}
	
\begin{scope}[xshift=145,local bounding box=prop3]
	\coordinate (XX) at (-1,0);
	\coordinate (YY) at (0,0);
	\coordinate (MM) at (-.5,0);
	
	\draw[ad](YY) -- (XX);
	\draw[fill = black] (XX) circle (1.5pt);
	\draw[fill = black] (YY) circle (1.5pt);
	
	\node [below=.1 of MM,scale=.9] {$\nu-\frac{k}{2\lambda}$};
	\node [left=.1 of XX,scale=1.1] {${\displaystyle \partial_y^{\alpha_1} \cdots \partial_y^{\alpha_k}}$};	
	\node [below=.1 of XX,scale=1.1] {\textit{x}};
	\node [below=.1 of YY,scale=1.1] {\textit{y}};
\end{scope}

\begin{scope}[xshift=195,local bounding box=eq]
	\coordinate (XX) at (0,0);	
	\node [above=-0.55 of XX,scale=1] {${\displaystyle\partial_y^{\alpha_1} \cdots \partial_y^{\alpha_k}\frac{1}{\|y-x\|^{2\lambda\nu-k}}}\ .$};
\end{scope}
	
	\path(prop1.east) -- (prop1-|prop2.west)  node[midway,above=.05]{=};
	\path(prop2.east) -- (prop2-|prop3.west)  node[midway,above=-.07]{=};
	\path(prop3.east) -- (prop3-|eq.west)  node[midway,above=-.03]{=};
\end{tikzpicture}
\end{equation}
The subscript $\nu\in\RR$ refers to the scaling weight $\|x\|^{-2\lambda\nu}$ of $Q^\alpha_\nu(x)$ and of $Q_{\nu;\alpha}(x)$. We say that the weight of the edge from $x$ to $y$ is  $\nu$,
\begin{equation}
N_{Q^\alpha_\nu}=N_{Q_{\nu;\alpha}}=\nu.
\end{equation}
A propagator may have indices which are repeated twice (but not more often). We use Einstein's sum convention that double indices are summed over (from 1 to $D$) without explicitly writing the sum symbol.
We also assume that $\alpha$ is only defined up to permutations, so that $\alpha$ becomes a multi-set with a maximum of two repetitions of each entry. Changing the orientation of a propagator line results in a minus sign for odd $|\alpha|$.

We work in Euclidean signature so that the metric tensor becomes a Kronecker delta,
$$
g^{\alpha_i,\alpha_j}=\delta_{\alpha_i,\alpha_j}.
$$
So, e.g., $g^{\alpha_1,\alpha_1}=\delta_{\alpha_1,\alpha_1}=D$ for every single index $\alpha_1$. For conceptual reasons, we use the symbol $g$ and not $\delta$ for the metric.

The propagator $Q_{\nu;\alpha}$ is connected to the momentum space propagator via a Fourier transformation
\begin{equation}
\int_{\RR^D}\frac{p^{\alpha_1}\cdots p^{\alpha_k}}{\|p\|^{2\mu}}\ee^{\ii x\cdot p}\frac{\dd^D p}{2^D\pi^{D/2}}=
\frac{\partial_{\alpha_1}\cdots\partial_{\alpha_k}}{\ii^k}\int_{\RR^D}\frac{\ee^{\ii x\cdot p}}{\|p\|^{2\mu}}\frac{\dd^D p}{2^D\pi^{D/2}}=
\frac{\Gamma(\lambda+1-\mu)}{4^\mu\ii^k\Gamma(\mu)}
Q_{1+\frac{1+k/2-\mu}\lambda;\alpha}(0,x).
\end{equation}
A chain of two propagators with multi-indices $\alpha$ and $\beta$ can hence be expressed as a single propagator,
\begin{align}\label{convolution}
&\int_{\RR^D}Q_{\nu;\alpha}(x,y)\,Q_{\mu;\beta}(y,z)\frac{\dd^D y}{\pi^{D/2}}=\nonumber\\
&\qquad\frac{\Gamma\Big(1+\frac{|\alpha|}2+\lambda(1-\nu)\Big)\Gamma\Big(1+\frac{|\beta|}2+\lambda(1-\mu)\Big)\Gamma\Big(\lambda(\nu+\mu-1)-1-\frac{|\alpha|+|\beta|}2\Big)}
{\Gamma\Big(2+\frac{|\alpha|+|\beta|}2+\lambda(2-\nu-\mu)\Big)\Gamma\Big(\lambda\nu-\frac{|\alpha|}2\Big)\Gamma\Big(\lambda\mu-\frac{|\beta|}2\Big)}
Q_{\nu+\mu-1-\frac1\lambda;\alpha,\beta}(x,z),
\end{align}
whenever the right-hand side exists. Note that we used differential propagators
to obtain a simple formula for the propagator chain.

In the numerator form (\ref{numform}), double edges can be combined to a single edge
\begin{equation}
Q_\nu^\alpha(x,y)\,Q_\mu^\beta(x,y)=Q_{\nu+\mu}^{\alpha\beta}(x,y)
\end{equation}
for any multi-labels $\alpha$ and $\beta$. Moreover, repeated indices can be dropped in
$Q_\nu^\alpha(x,y)$,
\begin{equation}
Q^{\alpha_1,\alpha_1,\alpha_2,\ldots,\alpha_k}_\nu(x,y)=Q^{\alpha_2,\ldots,\alpha_k}_\nu(x,y).
\end{equation}
In differential form, the reduction of double indices is more subtle. By explicit differentiation we get for $\alpha=\alpha_1,\alpha_1,\alpha_2,\ldots,\alpha_k$,
\begin{equation}\label{Qbb}
Q_{\nu;\alpha_1,\alpha_1,\alpha_2,\ldots,\alpha_k}(x,y)=(|\alpha|-2\lambda\nu)(|\alpha|-2\lambda(\nu-1))Q_{\nu;\alpha_2,\ldots,\alpha_k}(x,y),\quad\text{if }\nu\neq1+\frac{|\alpha|}{2\lambda}.
\end{equation}
For $\nu=1+|\alpha|/2\lambda$ we obtain a Dirac $\delta$ distribution,
\begin{equation}\label{Qdelta}
Q_{1+\frac{|\alpha|}{2\lambda};\alpha_1,\alpha_1,\alpha_2,\ldots,\alpha_k}(x,y)=-\frac4{\Gamma(\lambda)}\partial_y^{\alpha_2}\cdots\partial_y^{\alpha_k}\delta^{D}(x-y).
\end{equation}
Upon integration, such an edge will contract.

The transition from $Q_{\nu;\alpha}$ to $Q_\nu^\alpha$ can be calculated by iterative application of partial differentiation. For any single index $\beta$ and
$\alpha=\alpha_1,\ldots,\alpha_k$ we obtain
\begin{equation}\label{partialQ}
\partial^\beta_y Q^\alpha_\nu(x,y)=\sum_{i=1}^k\delta_{\beta,\alpha_i}Q^{\alpha_1,\ldots,\alpha_{i-1},\alpha_{i+1},\ldots,\alpha_k}_{\nu+\frac1{2\lambda}}(x,y)
-(2\lambda\nu+k)Q^{\beta,\alpha}_{\nu+\frac1{2\lambda}}(x,y).
\end{equation}
The left-hand side of (\ref{partialQ}) may be considered as a mixed numerator differential propagator $Q^\alpha_{\nu+1/2\lambda;\beta}$. Solving (\ref{partialQ})
for the last term on the right-hand side lowers the number of upper indices.
This gives rise to a bootstrap algorithm for the conversion from numerator form to differential form.
We arrange all spin $k-2j$, $j=0,1,\ldots,\lfloor k/2\rfloor$ propagators built from subsets of indices in $\alpha$ into a vectors
$Q^k_\nu$ and $Q_{\nu,k}$ for numerator and differential forms, respectively. Then, the transformation matrix $T^k_\nu$,
\begin{equation}
Q_{\nu,k}=T^k_\nu Q^k_\nu,
\end{equation}
between the vectors $Q^k_\nu$ and $Q_{\nu,k}$ is triangular (for a suitable arrangement of the entries).
The diagonal entries of $T^k_\nu$ are $(-2)^{k-2j}\Gamma(\lambda\nu+k/2-j)/\Gamma(\lambda\nu-k/2+j)$ for $j=0,\ldots,\lfloor k/2 \rfloor$.

\begin{ex}\label{ex1}
For $k=0$, we trivially have $T^0_\nu=(1)$.
For $k=1$ (see $k=0$ in (\ref{partialQ})), we have $Q_{\nu;\alpha_1}(x,y)=(-2\lambda\nu+1)Q^{\alpha_1}_\nu(x,y)$, yielding $T^1_\nu=(-2\lambda\nu+1)$.
For $k=2$ we have
$$
Q_{\nu;\alpha_1,\alpha_2}(x,y)=4\lambda\nu(\lambda\nu-1)Q^{\alpha_1,\alpha_2}_\nu(x,y)-2(\lambda\nu-1)g^{\alpha_1,\alpha_2}Q_\nu(x,y)
$$
leading to the transition matrix
$$
T^2_\nu=\left(\begin{array}{cc}4\lambda\nu(\lambda\nu-1)&
-2(\lambda\nu-1)g^{\alpha_1,\alpha_2}\\0&1\end{array}\right).
$$
\end{ex}

Let $0$ be the origin in $\RR^D$ and let $\hat z_1$ be a fixed, constant unit vector, $\|\hat z_1\|=1$. We obtain $Q^\alpha_\nu(0,\hat z_1)=(-1)^{|\alpha|}Q^\alpha_\nu(\hat z_1,0)=\hat z_1^{\alpha_1}\cdots\hat z_1^{\alpha_k}$.
For any vector $x\in\RR^D$, we use the anti-symmetry of the propagator to define $Q_{\nu;\alpha}(x,0)=(-1)^{|\alpha|}Q_{\nu;\alpha}(0,x)$ and $Q_{\nu;\alpha}(x,\hat z_1)=(-1)^{|\alpha|}Q_{\nu;\alpha}(\hat z_1,x)$.
For the definition of $Q_{\nu;\alpha}(0,\hat z_1)=(-1)^{|\alpha|}Q_{\nu;\alpha}(\hat z_1,0)$ we use the transition matrix $T^k_\nu$. In general, both propagators depend on $\nu$.

\begin{ex}\label{ex2}
From Example \ref{ex1} we obtain $Q_{\nu;\alpha_1}(0,\hat z_1)=(-2\lambda\nu+1)\hat z_1^{\alpha_1}$ and $Q_{\nu;\alpha_1,\alpha_2}(0,\hat z_1)=4\lambda\nu(\lambda\nu-1)\hat z^{\alpha_1}\hat z^{\alpha_2}-2(\lambda\nu-1)g^{{\alpha_1},{\alpha_2}}$.
\end{ex}

\section{Feynman integrals}
Consider an oriented graph $G$ whose edges are labeled by the same indices as the propagators (weight and spin).
Let $\VGext$ be a set of external vertices $z_0$, $z_1$, \ldots, $z_{|\VGext|-1}$ and $\VGint$ be a set of internal vertices $x_1$, \ldots, $x_{|\VGint|}$. We define the
Feynman integral $A_G(z_0,\ldots,z_{|\VGext|-1})$ as the integral
\begin{equation}\label{Adef}
A_G^\alpha(z_0,\ldots,z_{|\VGext|-1})=\int_{\RR^D}\frac{\dd^D x_1}{\pi^{D/2}}\cdots\int_{\RR^D}\frac{\dd^D x_\VGint}{\pi^{D/2}}\prod_{xy\in E_G}Q_{xy}(x,y),
\end{equation}
where the product is over all edges $e=xy\in E_G$. For each edge we have $Q_e=Q^{\alpha_e}_{\nu_e}$ for an edge in numerator form and $Q_e=Q_{\nu_e;\alpha_e}$ for an edge in differential form.
We assume that $G$ is such that the integral exists. Typically, on the right-hand side some (or all) indices are contracted, so that the indices of the propagators
form a larger set than $\alpha$.

The scaling weight of the graph $G$ in $D=2\lambda+2$ dimensions is the superficial degree of divergence
\begin{equation}\label{NGdef}
N_G=\sum_{e\in E_G}\nu_e-\frac{\lambda+1}\lambda\VGint.
\end{equation}
The existence of (\ref{Adef}) only depends on the scaling weights of $G$ and its
subgraphs, so that we can refer to spin zero, Proposition 11 in \cite{gfe}, for convergence.

\section{One-point integrals}\label{sectper}
Consider a graph $G$ with one external vertex $0\in\RR^D$ (also labeled $0$) and zero scaling weight
\begin{equation}\label{NG0}
N_G=0.
\end{equation}
For the definition of the Feynman integral of $G$ we pick any vertex $z_1\neq0$ in $G$ (leading to the graph $G_{z_1}$ with two external vertices) and integrate the two-point
function $A_{G_{z_1}}^\alpha(0,z_1)$ over the $(D-1)$-dimensional unit-sphere parametrized by $\hat z_1=z_1/\|z_1\|$,
\begin{equation}\label{Pdef}
P_G^\alpha=\|z_1\|^D\int_{S^{D-1}}A_{G_{z_1}}^\alpha(0,z_1)\dd^{D-1}\hat z_1,
\end{equation}
where the integral is normalized to $\int_{S^{D-1}}\dd^{D-1} \hat z_1=1$. By rotational symmetry, the integral over $\hat z$ is trivial in the scalar case $\alpha=\emptyset$.
If, moreover, the dimension $D$ is an integer, then $P_G\in\RR$ is a primitive scalar Feynman integral which is a mathematical period in the sense of \cite{KZ}.
The Feynman integral $P_G^\emptyset$ can also be defined in non-integer dimensions $D=2n+4-\epsilon$, $n=0,1,2,\ldots$.
In this setup, every Laurent coefficient in $\epsilon$ is a mathematical period \cite{BWperiods}.

The one-point integral $P_G$ depends neither on the choice of $0$ and $z_1$ nor
on $\|z_1\|$.
\begin{lem}\label{lem:period}
Let $G$ be a graph with an external label $0$ and $N_G=0$.
Then, the integral (\ref{Pdef}) does not depend on the choices of $0$ and $z_1$.

That is, $P_G^\alpha$ is well defined and for every $G'$ which is $G$ with a different external label $0$ ($0$ in $G$ is internal in $G'$ and vice versa) we have
$N_{G'}=0$ and
\begin{equation}\label{PGG'}
P_G^\alpha=P_{G'}^\alpha.
\end{equation}
Moreover, $P_G^\alpha=0$ if the total spin $|\alpha|$ of $G$ is odd.
\end{lem}
\begin{proof}
It is obvious from the definition (\ref{NGdef}) that $N_{G'}=0$.

We first show the independence of $z_1$. The total scaling weight of $A_{G_{z_1}}^\alpha(0,z_1)$ is $\|z_1\|^{-2\lambda(N_G+(\lambda+1)/\lambda)}$ (because the graph $G_{z_1}$
loses the internal vertex corresponding to $z_1$). We obtain independence from $\|z_1\|$ from (\ref{NG0}) and (\ref{Dl}). Independence from the orientation of $z_1$ is explicit by
integration over the unit sphere.

Next, we prove the independence of the choice of $z_1$ in $G$. We may assume that $z_1=\hat z_1$ is a unit vector. We pick an internal vertex in $G_{\hat z_1}$ and label it $x_1$. Next, we scale all internal vertices $\neq x_1$ by $x_i\mapsto x_i\|x_1\|$. From each edge $e$ we extract the scaling weight $\|x_1\|^{-2\lambda\nu_e}$ of the propagator $Q_e$.
We collect a total weight $\|x_1\|^{-2D}$ because $N_G=0$ and we are not scaling $\hat z_1$ and $x_1$. Moreover, the scaling replaces $x_1$ with the unit vector $\hat x_1=x_1/\|x_1\|$
and $\hat z_1$ with $\hat z_1/\|x_1\|$. We have integrations over $\hat x_1$, $\hat z_1$, and
$\|x_1\|$. We swap the variables $\hat x_1$ and $\hat z_1$ and invert $\|x_1\|$ by $\|x_1\|\mapsto\|x_1\|^{-1}$. The scaling weight $\|x_1\|^{-2D}$ compensates for the change in the integration measure.
We combine the integral over the new $\hat x_1$ with the integral over $\|x_1\|$ to the $D$-dimensional integral over $x_1$ and the claim follows.

A transformation of all internal integration variables $x_i\mapsto -x_i$ maps $z_1$ to $-z_1$. We pick up a sign $(-1)^{\sum_e|\alpha_e|}=(-1)^{|\alpha|}$ in the integral (\ref{Pdef}).
Independence of $z_1$ gives $P_G^\alpha=0$ if $|\alpha|$ is odd.

A transformation of all internal integration variables $x_i\mapsto z_1-x_i$ swaps $0$ and $z_1$ with a sign $(-1)^{|\alpha|}$. If $|\alpha|$ is even, then (\ref{PGG'}) is proved.
If $|\alpha|$ is odd, then (\ref{PGG'}) is trivial.
\end{proof}

Assume that the spin $|\alpha|$ is even (otherwise $P_G^\alpha=0$). Let $\pi=\pi(\alpha)$ be a partition of $\{\alpha_1,\ldots,\alpha_{|\alpha|}\}$ into pairs $\{\pi_{i1},\pi_{i2}\}$, $i=1,\ldots,|\pi|=|\alpha|/2$.
Let $\Pi_0^\alpha$ be the set of all such partitions (later we will define $\Pi_1^\alpha$ and $\Pi_2^\alpha$ for 2- and 3-point functions). Note that the set $\Pi_0^\alpha$ has
$|\Pi_0^\alpha|=(|\alpha|-1)!!=|\alpha|!/(2^{|\alpha|/2}(|\alpha|/2)!)$ elements.

Because $P_G^\alpha$ does not depend on any external vectors, it is a sum over all possibilities to construct
a spin $|\alpha|$ vector from products of $g^{\alpha_i,\alpha_j}$. For $|\alpha|\geq2$ we obtain
\begin{equation}\label{PGP}
P_G^\alpha=\sum_{\pi\in\Pi_0^\alpha}P_{G,\pi} g^{\pi(\alpha)},
\end{equation}
where we used the notation
\begin{equation}
g^\pi=g^{\pi_{11},\pi_{12}}\cdots g^{\pi_{N1},\pi_{N2}}
\end{equation}
with $N=|\alpha|/2$. For $\alpha=\emptyset$ we set $P_{G,\emptyset}=P_G$ and $g^\emptyset=1$.

By contraction over repeated indices, $\pi\mapsto g^\pi$ defines a bilinear form on  $\Pi_0^\alpha$,
\begin{equation}\label{bilin}
\langle \pi_1,\pi_2\rangle_0=g^{\pi_1}g^{\pi_2}\in\ZZ[D],\qquad\text{for }\pi_1,\pi_2\in\Pi_0^\alpha.
\end{equation}

\begin{lem}\label{lemnondeg1}
The bilinear form $\langle\cdot,\cdot\rangle_0$ is nondegenerate.
\end{lem}
\begin{proof}
We consider $\{1,\ldots,|\alpha|\}$ as the vertices of a graph $\Gamma$ and every pair in $\pi\in\Pi_0^\alpha$ as an edge in $\Gamma$. Then $\Gamma$ has
$|\alpha|/2$ disconnected edges ($\Gamma$ is a 1-regular graph).
For $\pi_1,\pi_2\in\Pi_0^\alpha$, the graph $\pi_1\cup\pi_2$ is a 2-regular graph. Any such graph is a collection of $L(\pi_1,\pi_2)$ loops. By contraction of indices, $\langle\pi_1,\pi_2\rangle_0=D^{L(\pi_1,\pi_2)}$. In particular, $\langle\pi,\pi\rangle_0=D^{|\alpha|/2}$ whereas for any $\pi_1\neq\pi_2$ we obtain that
$\langle\pi_1,\pi_2\rangle_0=D^n$ with $n<|\alpha|/2$. Therefore, $\det\langle\pi_1,\pi_2\rangle_0$ has unit leading coefficient as a polynomial in $D$. In particular, $\det(\langle\pi_1,\pi_2\rangle_0)\neq0$
and $\langle\cdot,\cdot\rangle_0$ is nondegenerate.
\end{proof}
If one (artificially) divides the $|\alpha|$ vertices of the graph $\Gamma$ in the above proof into $|\alpha|/2$ upper vertices and $|\alpha|/2$ lower vertices, one obtains an element in the Brauer algebra \cite{cell}.
The scalar product (\ref{bilin}) is the trace in the Brauer algebra. For generic parameter $D$, Brauer algebras have been shown to be semi-simple. Brauer algebras
belong to the class of cellular algebras.

By Lemma \ref{lemnondeg1}, every partition $\pi\in\Pi_0^\alpha$ has a dual $\hat\pi=\hat\pi(\alpha)$ in the vector space of formal sums of partitions with coefficients in the field of rational functions in $D$,
\begin{equation}
\langle\pi_i,\hat\pi_j\rangle_0=\delta_{i,j},\qquad\pi_i\in\Pi_0^\alpha,\;\hat\pi_j\in\langle\Pi_0^\alpha\rangle_{\QQ(D)}.
\end{equation}
By linearity, we extend $\hat\pi$ to $g^{\hat\pi}$ yielding
\begin{equation}
g^{\pi_i}g^{\hat\pi_j}=\delta_{i,j}.
\end{equation}
From (\ref{PGP}) we hence obtain
\begin{equation}\label{PGpi}
P_{G,\pi}=P_G^\alpha g^{\hat\pi(\alpha)},\quad\text{for }\pi\in\Pi_0^\alpha.
\end{equation}

\begin{ex}\label{ex1pt}
For $|\alpha|=2$ we write $12$ for the pair $\alpha_1,\alpha_2$ and get $\Pi_0^\alpha=\{\{12\}\}$. The dual of $\{12\}$ is $\frac1D\{12\}$. Hence
$$
P_{G,\{12\}}=\frac{P_G^{\alpha_1,\alpha_2}g^{\alpha_1,\alpha_2}}{D}.
$$
For $|\alpha|=4$ we have $\Pi_0^\alpha=\{\{12,34\},\{13,24\},\{14,23\}\}=\{\pi_1,\pi_2,\pi_3\}$. A short calculation gives
$$
\hat\pi_1=\frac{(D+1)\pi_1-\pi_2-\pi_3}{(D-1)D(D+2)}
$$
with cyclic results for $\hat\pi_2$ and $\hat\pi_3$. Hence
$$
P_{G,\{12,34\}}=\frac{P_G^{\alpha_1,\alpha_2,\alpha_3,\alpha_4}((D+1)g^{\alpha_1,\alpha_2}g^{\alpha_3,\alpha_4}-g^{\alpha_1,\alpha_3}g^{\alpha_2,\alpha_4}
-g^{\alpha_1,\alpha_4}g^{\alpha_2,\alpha_3})}{(D-1)D(D+2)},
$$
with cyclic results for $P_{G,\{13,24\}}$ and $P_{G,\{14,23\}}$.

Results up to spin $|\alpha|=10$ are in the Maple package {\tt HyperlogProcedures} \cite{Shlog}.
\end{ex}
For fixed $G$, we lift duality from partitions to formal sums of (spin zero) graphs in the graph algebra with coefficients in $\QQ(D)$.
For $\pi\in\Pi_0^\alpha$ we denote the lift of $\hat\pi=\sum_ic_i\pi_i$ ($c_i\in\QQ(D)$) by $G_{\hat\pi}$, i.e.\ $G_{\hat\pi}=\sum_ic_iGg^{\pi_i}$.
We define the Feynman integral of $G_{\hat\pi}$ as
\begin{equation}\label{PGpihat0}
P_{G_{\hat\pi}}=\sum_ic_iP_{Gg^{\pi_i}}=\sum_ic_iP_G^\alpha g^{\pi_i(\alpha)}=P_G^\alpha g^{\hat\pi(\alpha)}=P_{G,\pi},\quad\text{if }\hat\pi=\sum_ic_i\pi_i,
\end{equation}
where in the last equation we used (\ref{PGpi}). With (\ref{PGP}) we get
\begin{equation}\label{PGP1}
P_G^\alpha=\sum_{\pi\in\Pi_0^\alpha}P_{G_{\hat\pi}}g^\pi.
\end{equation}
Hence, spin $|\alpha|$ one-point integrals can be expressed as sums of scalar Feynman integrals.
\begin{ex}\label{ex1pta}
From Example \ref{ex1pt} we obtain
\begin{align*}
G_{\widehat{\{12\}}}&=\frac{G(\alpha)g^{\alpha_1,\alpha_2}}D,\\
G_{\widehat{\{12,34\}}}&=\frac{(D+1)G(\alpha)g^{\alpha_1,\alpha_2}g^{\alpha_3,\alpha_4}-G(\alpha)g^{\alpha_1,\alpha_3}g^{\alpha_2,\alpha_4}
-G(\alpha)g^{\alpha_1,\alpha_4}g^{\alpha_2,\alpha_3}}{(D-1)D(D+2)},
\end{align*}
where $G(\alpha)$ indicates that the graph $G$ depends on the spin indices $\alpha$. On the right-hand sides all spin indices are contracted, so that the graphs are scalar.
\end{ex}

\section{Two-point functions}\label{secttwopt}
A two-point function has two external vertices $0=z_0$ and $z_1$. By scaling internal variables we obtain
\begin{equation}\label{AGNG}
A_G^\alpha(0,z_1)=\|z_1\|^{-2\lambda N_G}A_G^\alpha(0,\hat z_1),\quad\text{where }\hat z_1=z_1/\|z_1\|.
\end{equation}
A transformation $x_i\mapsto\hat z_1-x_i$ for all internal vertices gives
\begin{equation}
A_G^\alpha(\hat z_1,0)=(-1)^{|\alpha|}A_G^\alpha(0,\hat z_1).
\end{equation}

The Feynman integral $A_G^\alpha(0,\hat z_1)$ is a linear combination of products of $\hat z_1$ and $g$. To express this linear combination, we use a partition
of the set $\{\alpha_1,\ldots,\alpha_{|\alpha|}\}$ into $\pi^0_1,\ldots,\pi^0_N,\pi^1$ where the sets $\pi^0_i=\{\pi^0_{i1},\pi^0_{i2}\}$ are pairs.
The last slot $\pi^1$ may have any number of elements. We order $\pi^0$ before $\pi^1$, so that, e.g., in the case $|\alpha|=4$ we distinguish the partitions
$\pi^0=\{12\},\pi^1=34$ and $\pi^0=\{34\},\pi^1=12$. (To lighten the notation we omit brackets for sets of labels.)
Let $\Pi_1^\alpha$ be the set of all these partitions.
The set $\Pi_0^\alpha$ in the previous section corresponds to the subset of $\Pi_1^\alpha$ with empty last slot. The set $\Pi_1^\alpha$ has
\begin{equation}
|\Pi_1^\alpha|=1+\frac1{1!}\genfrac(){0pt}{}{|\alpha|}{2}+\frac1{2!}\genfrac(){0pt}{}{|\alpha|}{2,2}+\ldots=\sum_{j=0}^{\lfloor\frac{|\alpha|}2\rfloor}\frac{|\alpha|!}{2^j(|\alpha|-2j)!j!}
\end{equation}
elements. For $\pi\in\Pi_1^\alpha$ we use the shorthand
\begin{equation}\label{gzpi}
g^{\pi^0}\hat z_1^{\pi^1}=g^{\pi^0_{11},\pi^0_{12}}\cdots g^{\pi^0_{N1},\pi^0_{N1}}\hat z_1^{\pi^1_1}\cdots\hat z_1^{\pi^1_{|\alpha|-2N}},\quad N=|\pi^0|,
\end{equation}
for the corresponding expansion into products of $g$ and $\hat z_1$.
With (\ref{AGNG}) we get
\begin{equation}\label{AGP}
A_G^\alpha(0,z_1)=\|z_1\|^{-2\lambda N_G}\sum_{\pi\in\Pi_1^\alpha}P_{G,\pi}\,g^{\pi^0}\hat z_1^{\pi^1}.
\end{equation}

The core information of the Feynman integral $A_G^\alpha(0,z_1)$ is encoded in the numbers $P_{G,\pi}$, which
are the $g^{\pi^0}\hat z_1^{\pi^1}$ coefficients of $A_G^\alpha(0,\hat z_1)$.

If $|\alpha|$ is even, integration over $\hat z_1$ connects $P_{G,\pi}$ to the one-point integral of the graph $G_0$ which is $G$ with additional edge $01$ of weight
\begin{equation}\label{nu01}
\nu_{01}=-N_G+\frac\lambda{\lambda+1}.
\end{equation}
Note that $\nu_{01}$ is chosen such that $N(G_0)=0$. If $|\alpha|$ is odd, then $P_{G_0}^\alpha=0$.

\begin{ex}\label{ex2pt}
With the graph $G(\alpha)$ consisting of an edge $01$ with spin $\alpha$ (such that
$A_G^\alpha(0,\hat z_1)=\hat z_1^\alpha$) we get from Example \ref{ex1pta}
\begin{align}
\int_{S^{D-1}}\hat z_1^{\alpha_1}\hat z_1^{\alpha_2}\dd^{D-1}\hat z_1&=\frac1D,\\\nonumber
\int_{S^{D-1}}\hat z_1^{\alpha_1}\hat z_1^{\alpha_2}\hat z_1^{\alpha_3}\hat z_1^{\alpha_4}\dd^{D-1}\hat z_1
&=\frac{g^{\alpha_1,\alpha_2}g^{\alpha_3,\alpha_4}+g^{\alpha_1,\alpha_3}g^{\alpha_2,\alpha_4}+g^{\alpha_1,\alpha_4}g^{\alpha_2,\alpha_3}}{D(D+2)}.
\end{align}
We obtain for $|\alpha|=2$
\begin{equation}
P_{G_0,\{12\}}=P_{G,(\{12\},\emptyset)}+\frac{P_{G,(\emptyset,12)}}D,
\end{equation}
and for $|\alpha|=4$
\begin{equation}
P_{G_0,\{12,34\}}=P_{G,(\{12,34\},\emptyset)}+\frac{P_{G,(\{12\},34)}}D+\frac{P_{G,(\emptyset,1234)}}{D(D+2)}.
\end{equation}
\end{ex}

To calculate $P_{G,\pi}$ from spin zero Feynman integrals (corresponding to unlabeled scalar graphs), we proceed as in the one-point case and define a bilinear form on $\Pi_1^\alpha$,
\begin{equation}
\langle\pi_1,\pi_2\rangle_1=g^{\pi_1^0}\hat z_1^{\pi_1^1}g^{\pi_2^0}\hat z_1^{\pi_2^1}\in\ZZ[D],\quad\text{for }\pi_1,\pi_2\in\Pi_1^\alpha.
\end{equation}
Note that $\langle\pi_1,\pi_2\rangle_1$ does not depend on $\hat z_1$ because all indices are contracted and $\|\hat z_1\|=1$.
\begin{lem}\label{lemnondeg2}
The bilinear form $\langle\cdot,\cdot\rangle_1$ is nondegenerate.
\end{lem}
\begin{proof}
As in the proof of Lemma \ref{lemnondeg1} we consider $\{1,\ldots,|\alpha|\}$ as vertex set of a graph $\Gamma$.
For any $\pi\in\Pi_1^\alpha$ the edges of $\Gamma$ are the pairs of $\pi^0$, while the elements in $\pi^1$ are isolated vertices. We obtain a graph with vertex degree $\leq1$.
The graph of $\pi_1\cup\pi_2$ with $\pi_1,\pi_2\in\Pi_1^\alpha$ has vertex degree $\leq2$. Loops in $\pi_1\cup\pi_2$ contribute a factor of $D$ to the bilinear form,
whereas (after relabeling the indices) open strings contribute by
$$
\hat z_1^{\alpha_1}g^{\alpha_1,\alpha_2}g^{\alpha_2,\alpha_3}\cdots g^{\alpha_{i-1},\alpha_i}\hat z_1^{\alpha_i}=\hat z_1^{\alpha_1}\hat z_1^{\alpha_1}=\|\hat z_1\|^2=1.
$$
Let $L(\pi_1,\pi_2)$ be the number of loops in the graph of $\pi_1\cup\pi_2$. We get (as in the proof of Lemma \ref{lemnondeg1}) $\langle\pi_1,\pi_2\rangle_1=D^{L(\pi_1,\pi_2)}$.
In particular, $\langle\pi,\pi\rangle_1=D^{|\pi^0|}$.
If $\pi_1\neq\pi_2$ with $\pi_1^0\neq\emptyset\neq\pi_2^0$, the degree of $\langle\pi_1,\pi_2\rangle_1$ is strictly smaller than $\min(|\pi_1^0|,|\pi_2^0|)$.
If $\pi_1^0=\emptyset$ or $\pi_2^0=\emptyset$, the graph of $\pi_1\cup\pi_2$ has no loops and we get $\langle\pi_1,\pi_2\rangle_1=1$.
If we arrange this partition into the first slot, the matrix of $\langle \cdot,\cdot\rangle_1$ has first column and row 1 and in all other columns and rows strictly the highest powers of D on the
diagonal. Therefore, $\det(\langle\pi_1,\pi_2\rangle_1)$ has leading coefficient $1$ in $D$ and thus is non-zero. The claim follows.
\end{proof}

For even spin $|\alpha|$, one can proceed as in the one-point case and split the $|\alpha|$ vertices into $|\alpha|/2$ upper vertices and $|\alpha|/2$ lower vertices. This promotes every $\pi\in\Pi_1^\alpha$ to an element in the
rook-Brauer algebra which is also semisimple for generic parameter $D$ \cite{cell}.
The bilinear form $\langle\cdot,\cdot\rangle_1$ is the trace in the rook-Brauer algebra.

By Lemma \ref{lemnondeg2} every partition $\pi\in\Pi_1^\alpha$ has a dual $\hat\pi=\hat\pi(\alpha)$ in the vector space of formal sums of partitions with coefficients in $\QQ(D)$,
\begin{equation}
\langle\pi_i,\hat\pi_j\rangle_1=\delta_{i,j},\qquad\pi_i\in\Pi_1^\alpha,\;\hat\pi_j\in\langle\Pi_1^\alpha\rangle_{\QQ(D)}.
\end{equation}
By linearity we extend $\hat\pi$ to $g^{\hat\pi^0}z_1^{\hat\pi^1}$ yielding
\begin{equation}
g^{\pi_i^0}\hat z_1^{\pi_j^1}g^{\hat\pi_j^0}\hat z_1^{\hat\pi_j^1}=\delta_{i,j}.
\end{equation}
From (\ref{AGP}) we hence obtain
\begin{equation}\label{PGpi1}
P_{G,\pi}=A_G^\alpha(0,\hat z_1)\,g^{\hat\pi^0}\hat z_1^{\hat\pi^1}=A_G^\alpha(0,\hat z_1)\,g^{\hat\pi^0}Q_{\nu_{01}}^{\hat\pi^1}(0,\hat z_1),
\end{equation}
where we define $g^{\hat\pi^0}Q_{\nu_{01}}^{\hat\pi^1}(0,\hat z_1)$ in analogy to $g^{\hat\pi^0}z_1^{\hat\pi^1}$ in (\ref{gzpi}).
We choose $\nu_{01}$ as in (\ref{nu01}) to ensure that $G\cup01$ is a one-point graph. More precisely, $G\,g^{\hat\pi^0}Q_{\nu_{01}}^{\hat\pi^1}(0,\hat z_1)$ is a linear combination
of graphs with total spin zero. In particular, the vertices $0$ and $1$ can be chosen freely when calculating its Feynman integral (see Lemma \ref{lem:period}) so that we can drop the labels
in $G\,g^{\hat\pi^0}Q_{\nu_{01}}^{\hat\pi^1}(0,\hat z_1)$. Extending the Feynman integral to the graph algebra, we obtain
(compare (\ref{PGpihat0}))
\begin{equation}\label{PGpihat}
P_{G,\pi}=P_{G_{\hat\pi}}.
\end{equation}
Substitution into (\ref{AGP}) gives the two-point function as sum over propagators with spin zero Feynman integral coefficients,
\begin{equation}\label{AGPhat}
A_G^\alpha(0,z_1)=\|z_1\|^{-2\lambda N(G)}\sum_{\pi\in\Pi_1^\alpha}P_{G_{\hat\pi}}\,g^{\pi^0}\hat z_1^{\pi^1}=\sum_{\pi\in\Pi_1^\alpha}P_{G_{\hat\pi}}\,g^{\pi^0}Q^{\pi^1}_{N(G)}(0,z_1).
\end{equation}
With this formula one can eliminate any two-point insertion in a Feynman integral by a sum over propagators in numerator form (\ref{numform}) with coefficients which
are a product of a scalar Feynman integral and metric tensors.
\begin{ex}\label{ex2pta}
We write $G(\alpha)g^{\pi^0(\alpha)}Q_{\nu_{01}}^{\pi^1(\alpha)}$ for the spin zero graph that is $G(\alpha)$ contracted with $g^{\pi^0(\alpha)}$
and edge $01$ of spin $\pi^1(\alpha)$ and weight $\nu_{01}$ as in (\ref{nu01}).
For $|\alpha|=1$, we get $G_{\widehat{(\emptyset,1)}}=G(\alpha_1)Q_{\nu_{01}}^{\alpha_1}$.

For $|\alpha|=2$ we have $\Pi_1^\alpha=\{(\{12\},\emptyset),(\emptyset,12)\}$ and
$$
G_{\widehat{(\{12\},\emptyset)}}=\frac{G(\alpha)g^{\alpha_1,\alpha_2}Q_{\nu_{01}}-G(\alpha)Q_{\nu_{01}}^{\alpha_1,\alpha_2}}{D-1},\qquad
G_{\widehat{(\emptyset,12)}}=\frac{-G(\alpha)g^{\alpha_1,\alpha_2}Q_{\nu_{01}}
+D\,G(\alpha)Q_{\nu_{01}}^{\alpha_1,\alpha_2}}{D-1}.
$$
For $|\alpha|=3$ we have $\Pi_1^\alpha=\{(\{12\},3),(\{13\},2),(\{23\},1),(\emptyset,123)\}$. We obtain
\begin{align*}
G_{\widehat{(\{12\},3)}}&=\frac{G(\alpha)g^{\alpha_1,\alpha_2}Q_{\nu_{01}}^{\alpha_3}-G(\alpha)Q_{\nu_{01}}^{\alpha_1,\alpha_2,\alpha_3}}{D-1},\\
G_{\widehat{(\emptyset,123)}}&=\frac{-G(\alpha)g^{\alpha_1,\alpha_2}Q_{\nu_{01}}^{\alpha_3}-G(\alpha)g^{\alpha_1,\alpha_3}Q_{\nu_{01}}^{\alpha_2}-G(\alpha)g^{\alpha_2,\alpha_3}Q_{\nu_{01}}^{\alpha_1}
+(D+2)G(\alpha)Q_{\nu_{01}}^{\alpha_1,\alpha_2,\alpha_3}}{D-1},\\
\end{align*}
plus two cyclic permutations of the first equation. Results up to $|\alpha|=10$ are in \cite{Shlog}.
\end{ex}

\section{Three-point functions and the definition of graphical functions}\label{sectthreept}
A three-point function has three external vertices $0=z_0$, $z_1$, and $z_2$. By scaling all internal variables we obtain
\begin{equation}
A_G^\alpha(0,z_1,z_2)=\|z_1\|^{-2\lambda N_G}A_G^\alpha(0,z_1/\|z_1\|,z_2/\|z_1\|).
\end{equation}
To define the graphical function of $G$ we use the coordinates (in a suitably rotated coordinate frame)
\begin{equation}\label{eqzdef}
\hat z_1=\frac{z_1}{\|z_1\|}=\left(\begin{array}{c}1\\0\\0\\\vdots\\0\end{array}\right),\qquad
\hat z_2=\frac{z_2}{\|z_1\|}=\left(\begin{array}{c}\Re z\\\Im z\\0\\\vdots\\0\end{array}\right).
\end{equation}
Note that $\hat z_2$ is normalized by the length of $z_1$ and hence not a unit vector in general.
With (\ref{eqzdef}) we obtain
\begin{equation}\label{fdef0}
f_G^\alpha(z)=A_G^\alpha(0,\hat z_1,\hat z_2).
\end{equation}
Alternatively, we may express the invariants
\begin{equation}\label{eqinvs}
\frac{\|z_2-z_0\|^2}{\|z_1-z_0\|^2}=z\zz,\qquad\frac{\|z_2-z_1\|^2}{\|z_1-z_0\|^2}=(z-1)(\zz-1)
\end{equation}
in terms of the complex variable $z$ and its complex conjugate $\zz$. With these identifications we can equivalently define the graphical function of $G$ as the function $f^\alpha_G(z)$ that fulfills the equation
\begin{equation}\label{fdef}
A_G^\alpha(0,z_1,z_2)=\|z_1\|^{-2\lambda N_G}f_G^\alpha(z).
\end{equation}

The Feynman integral $A_G^\alpha(0,\hat z_1,\hat z_2)$ is a linear combination of products of the metric $g$ and the vectors $\hat z_1$, $\hat z_2$. To express this linear combination, we proceed in analogy
to the previous sections and define a partition of the spin index set
$\alpha=\{\alpha_1,\ldots,\alpha_{|\alpha|}\}$ into $\pi^0_1,\ldots,\pi^0_N$, $\pi^1$, $\pi^2$, where the sets $\pi^0_i=\{\pi^0_{i1},\pi^0_{i2}\}$ are pairs.
The slots $\pi^1$ and $\pi^2$ may have any number of elements. We always distinguish between $\pi^0$, $\pi^1$, and $\pi^2$, see Section \ref{secttwopt}.
Let $\Pi_2^\alpha$ be the set of all these partitions. The set $\Pi_2^\alpha$ has 
\begin{equation}
|\Pi_2^\alpha|=2^{|\alpha|}+\frac1{1!}\genfrac(){0pt}{}{|\alpha|}{2}2^{|\alpha|-2}+\frac1{2!}\genfrac(){0pt}{}{|\alpha|}{2,2}2^{|\alpha|-4}+\ldots=\sum_{j=0}^{\lfloor\frac {|\alpha|}2\rfloor}\frac{|\alpha|!}{(|\alpha|-2j)!j!}2^{|\alpha|-3j}
\end{equation}
elements. For $\pi\in\Pi_2^\alpha$ we use the shorthand
\begin{equation}
g^{\pi^0}\hat z_1^{\pi^1}\hat z_2^{\pi^2}=g^{\pi^0_{11},\pi^0_{12}}\cdots g^{\pi^0_{N1},\pi^0_{N2}}\hat z_1^{\pi^1_1}\cdots\hat z_1^{\pi^1_m}\hat z_2^{\pi^2_1}\cdots\hat z_2^{\pi^2_n}
\end{equation}
for the corresponding expansion into products of $g$, $\hat z_1$, and $\hat z_2$ (where $2N+m+n=|\alpha|$).

Using (\ref{eqzdef}), we can express the scalar products $\hat z_i^{\alpha_1}\hat z_j^{\alpha_1}$, $i,j\in\{1,2\}$ in terms of $z$ and $\zz$,
\begin{equation}\label{zij}
\hat z_1^{\alpha_1}\hat z_1^{\alpha_1}=1,\quad\hat z_1^{\alpha_1}\hat z_2^{\alpha_1}=\frac{z+\zz}2,\quad\hat z_2^{\alpha_1}\hat z_2^{\alpha_1}=z\zz.
\end{equation}

With the above notation we obtain
\begin{equation}\label{fGcomps}
f_G^\alpha(z)=\sum_{\pi\in\Pi_2^\alpha}f_{G,\pi}(z)g^{\pi^0}\hat z_1^{\pi^1}\hat z_2^{\pi^2},
\end{equation}
with an analogous expansion for $A_G^\alpha(0,z_1,z_2)$ from Equation (\ref{fdef}).

In the following we often consider the graphical function $f_G^\alpha(z)$ as a vector with components $f_{G,\pi}(z)$,
\begin{equation}\label{FG}
f_G^\alpha(z)\leftrightarrow(f_{G,\pi}(z))_{\pi\in\Pi_2^\alpha}.
\end{equation}
Note that $f_G^\alpha(z)$ depends via $\hat z_1$, $\hat z_2$ on the orientation of the coordinate system,
while the vector $(f_{G,\pi}(z))_{\pi\in\Pi_2^\alpha}$ does not. We can use any
orientation of the coordinate system (and in particular (\ref{eqzdef})) to determine the scalar functions $f_{G,\pi}(z)$.

It is possible to follow the previous sections and express $f_G^\alpha(z)$ in terms of spin zero graphical functions by dualizing $\pi\in\Pi_2^\alpha$ and lifting $\hat\pi$ to the
graph algebra with coefficients in $\QQ(D,z,\zz)$. In particular, we define the bilinear form
\begin{equation}
\langle\pi_1,\pi_2\rangle_2=g^{\pi_1^0}\hat z_1^{\pi_1^1}\hat z_2^{\pi_1^2}g^{\pi_2^0}\hat z_1^{\pi_2^1}\hat z_2^{\pi_2^2}\in\QQ[D,z,\zz],\quad\text{for }\pi_1,\pi_2\in\Pi_2^\alpha.
\end{equation}
\begin{lem}
The bilinear form $\langle\cdot,\cdot\rangle_2$ is nondegenerate.
\end{lem}
\begin{proof}
We proceed as in the proof of Lemma \ref{lemnondeg2}. Open strings may now terminate at $\hat z_1$ or at
$\hat z_2$. With (\ref{zij}) their value depends on $z$ and $\zz$ but not on the dimension $D$.
This does not affect the argument in the proof of Lemma \ref{lemnondeg2}.
\end{proof}

\begin{ex}
For $|\alpha|=1$ we obtain $\Pi_2^\alpha=\{(\emptyset,1,\emptyset),(\emptyset,\emptyset,1)\}$. For the vectors $\hat z_1$, $\hat z_2$ the Gram matrix $(\hat z_i\cdot\hat z_j)$ is
(see Equation (\ref{zij}))
\begin{equation}\label{zz}
\left(\begin{array}{cc}1&\frac{z+\zz}2\\
\frac{z+\zz}2&z\zz\end{array}\right)
\end{equation}
with determinant $-(z-\zz)^2/4$.
We read off the duals of $\Pi_2^\alpha$ lifted to the graph algebra (see Example \ref{ex2pta}),
\begin{align*}
G_{\widehat{(\emptyset,1,\emptyset)}}&=\frac{-4z\zz G(\alpha)\hat z_1^{\alpha_1}+2(z+\zz)G(\alpha)\hat z_2^{\alpha_1}}{(z-\zz)^2},\\
G_{\widehat{(\emptyset,\emptyset,1)}}&=\frac{2(z+\zz)G(\alpha)\hat z_1^{\alpha_1}-4G(\alpha)\hat z_2^{\alpha_1}}{(z-\zz)^2},
\end{align*}
where we write $G(\alpha)\hat z_i^\alpha$ for the graph $G$ with an additional edge with propagator $\hat z_i^\alpha$ (i.e.\ an edge $01$ or $0z$ with weight $-1/2\lambda$).
\end{ex}
Due to the dependence on $z$ and $\zz$, the calculation of the duals is significantly more complicated
than in the case of $\Pi^\alpha_0$ and $\Pi^\alpha_1$. In {\tt HyperlogProcedures} only the results
up to $|\alpha|=5$ are calculated \cite{Shlog}. Empirically, we find that for any $\alpha$ the denominators
of the duals factorize into powers of $z-\zz$.

Note, however, that dualizing graphs is not as efficient as in the previous sections because one obtains
large sums and (more importantly) the spin zero expression cannot be expressed in terms of unlabeled graphs: There exists no identity that swaps internal and external vertices.
Moreover, in many cases the scalar graphical functions suffer from higher order singularities. Nevertheless, dualizing is an important tool for calculating small
graphical functions which cannot be calculated by other means (a kernel graphical function). The singularities in the individual terms can be handled with dimensional regularization.
Dualizing ensures that every scalar coefficient $f_{G,\pi}(z)$ inherits the general properties of scalar graphical functions.

In most cases, rather than expressing $f_G^\alpha(z)$ in terms of spin zero graphical functions, we try to construct $f_G^\alpha(z)$ from the empty graphical function (or a known kernel)
by the following five operations:
\begin{enumerate}
\item elimination of two-point insertions using Section \ref{secttwopt},
\item adding edges between external vertices,
\item permutation of external vertices,
\item product factorization,
\item appending an edge to the external vertex $z$.
\end{enumerate}
These operations are explained in the next sections.

\section{Identities for graphical functions}
\subsection{Edges between external vertices}\label{sectextedge}
Edges between external vertices correspond to constant factors in the Feynman integral.
The graphical function $f_G^\alpha(z)$ is multiplied by the propagator $Q^\beta_\nu(z_i,z_j)$ of the external edge $z_iz_j$.
The spin changes accordingly; contraction of indices lowers the spin, otherwise the spin increases. The vector $f_{G,\pi}(z)$ in (\ref{FG}) is multiplied by a rectangular matrix.

\begin{ex}\label{exextedge}
Consider an edge of weight $\nu$ and spin 1 from $1=z_1$ to $z=z_2$. The propagator of such an edge in the coordinates (\ref{eqzdef}) is
$$
Q^{\beta}_\nu(z)=\frac{\hat z_2^{\beta}-\hat z_1^{\beta}}{((z-1)(\zz-1))^{\lambda\nu+\frac12}}.
$$
Assume the graphical function $f_G^\alpha(z)$ also has spin 1,
$$
f_G^{\alpha_1}(z)=f_1(z)\hat z_1^{\alpha_1}+f_2(z)\hat z_2^{\alpha_1},
$$
such that $f_1(z)=f_{G,(\emptyset,1,\emptyset)}(z)$ and $f_2(z)=f_{G,(\emptyset,\emptyset,1)}(z)$. The vector $(f_1(z),f_2(z))^T$ has two components.
If $\beta=\alpha_1$, so that the index $\alpha_1$ is contracted, then (see (\ref{zij}))
$$
f_G^{\alpha_1}(z)Q^{\alpha_1}_\nu(z)=\frac{(z+\zz-2)f_1(z)+(2z\zz-z-\zz)f_2(z)}{2((z-1)(\zz-1))^{\lambda\nu+\frac12}}.
$$
The matrix of the multiplication is $1\times2$,
$$
\frac1{2((z-1)(\zz-1))^{\lambda\nu+\frac12}}\left(z+\zz-2,\quad 2z\zz-z-\zz\right).
$$
If $\beta=\alpha_2\neq\alpha_1$, we obtain the spin $2$ graphical function
$$
f_G^{\alpha_1}(z)Q^{\alpha_2}_\nu(z)=\frac{-f_1(z)\hat z_1^{\alpha_1}\hat z_1^{\alpha_2}+f_1(z)\hat z_1^{\alpha_1}\hat z_2^{\alpha_2}
-f_2(z)\hat z_2^{\alpha_1}\hat z_1^{\alpha_2}+f_2(z)\hat z_2^{\alpha_1}\hat z_2^{\alpha_2}}
{((z-1)(\zz-1))^{\lambda\nu+\frac12}}.
$$
The vector of the spin $2$ graphical function has the components $(\emptyset,12,\emptyset)$, $(\emptyset,1,2)$, $(\emptyset,2,1)$, $(\emptyset,\emptyset,12)$.
The component corresponding to the metric $g^{\alpha_1,\alpha_2}$ is zero, $f_{G,(\{12\},\emptyset,\emptyset)}=0$.
This leads to the $5\times2$ multiplication matrix
$$
\frac1{((z-1)(\zz-1))^{\lambda\nu+\frac12}}\left(\begin{array}{cc}0&0\\
-1&0\\
1&0\\
0&-1\\
0&1\end{array}\right).
$$
\end{ex}

\subsection{Permutation of external vertices}\label{sectpermute}
A transformation $x_i\mapsto z_1-x_i$ at all internal vertices gives
\begin{equation}
A_G^\alpha(0,z_1,z_2)=(-1)^{|\alpha|}A_G^\alpha(z_1,0,z_1-z_2).
\end{equation}
On the right-hand side, the positions of $z_0=0$ and $z_1$ are swapped. The transformation
$z_2\mapsto z_1-z_2$ implies $\hat z_2\mapsto\hat z_1-\hat z_2$ and, via the invariants (\ref{eqinvs}),
the map $z\mapsto1-z$. From (\ref{fdef}) we get the transformation of the graph $G=G_{01z}$ with external labels $0$, $1$, $z$,
\begin{equation}
f_{G_{01z}}^\alpha(z)=(-1)^{|\alpha|}f_{G_{10z}}^\alpha(1-z)=(-1)^{|\alpha|}f_{G_{10(1-z)}}^\alpha(z),
\end{equation}
where the last identity defines $f_{G_{10(1-z)}}^\alpha(z)$.

\begin{ex}
We continue Example \ref{exextedge} and consider a spin $1$ graphical function $f_G^{\alpha_1}(z)$.
Swapping $0$ and $1$ in the graph $G=G_{01z}$ gives
$$
-f_1(1-z)\hat z_1^\alpha-f_2(1-z)(\hat z_1^\alpha-\hat z_2^\alpha),
$$
which is represented by the $2\times2$ matrix
$$
\left(\begin{array}{cc}-1&-1\\
0&1\end{array}\right)
$$
together with the M\"obius transformation $z\mapsto1-z$ in the arguments of the components.
Note that the matrix squares to the identity, which reflects that the transformation is an involution.
\end{ex}

If we swap $z_1$ and $z_2$ in (\ref{fdef}) we get a factor $\|z_2\|^{-2\lambda N_G}$ and a transformation $z\mapsto z^{-1}$ from (\ref{eqinvs}).
Moreover, $\hat z_1^\alpha\mapsto z_2^\alpha/\|z_2\|=\hat z_2^\alpha/\|\hat z_2\|$ and $\hat z_2^\alpha\mapsto z_1^\alpha/\|z_2\|=\hat z_1^\alpha/\|\hat z_2\|$.
This implies that $\hat z_1$ and $\hat z_2$ are swapped together with an extra scaling factor $\|\hat z_2\|^{-1}=(z\zz)^{-1/2}$.

Altogether, we obtain a scale transformation, the inversion $z\mapsto z^{-1}$, and a permutation $\hat z_1\leftrightarrow \hat z_2$ in the spin structure.

\begin{ex}
For the spin $1$ graphical function $f_G^{\alpha_1}(z)$ we obtain from swapping labels $1$ and $z$ in the graph $G=G_{01z}$,
$$
(z\zz)^{-\lambda N_G-\frac12}(f_1(z^{-1})\hat z_2^\alpha+f_2(z^{-1})\hat z_1^\alpha).
$$
This corresponds to multiplication with the matrix
$$
(z\zz)^{-\lambda N_G-\frac12}\left(\begin{array}{cc}0&1\\
1&0\end{array}\right)
$$
together with the M\"obius transformation $z\mapsto1/z$ in the arguments of the components.
\end{ex}

\begin{ex}\label{expermext}
For a spin $2$ graphical function
\begin{equation}\label{eqpermext}
f_G^{\alpha_1,\alpha_2}(z)=f_0(z)g^{\alpha_1,\alpha_2}+f_1(z)\hat z_1^{\alpha_1}\hat z_1^{\alpha_2}+f_2(z)\hat z_1^{\alpha_1}\hat z_2^{\alpha_2}+f_3(z)\hat z_2^{\alpha_1}\hat z_1^{\alpha_2}+f_4(z)\hat z_2^{\alpha_1}\hat z_2^{\alpha_2}
\end{equation}
we obtain the transformation matrix
$$
(z\zz)^{-\lambda N_G-1}(z\zz f_0(z^{-1})g^{\alpha_1,\alpha_2}+f_1(z^{-1})\hat z_2^{\alpha_1}\hat z_2^{\alpha_2}+f_2(z^{-1})\hat z_2^{\alpha_1}\hat z_1^{\alpha_2}
+f_3(z^{-1})\hat z_1^{\alpha_1}\hat z_2^{\alpha_2}+f_4(z^{-1})\hat z_1^{\alpha_1}\hat z_1^{\alpha_2},
$$
corresponding to the matrix
$$
(z\zz)^{-\lambda N_G-1}\left(\begin{array}{ccccc}z\zz&0&0&0&0\\
0&0&0&0&1\\0&0&0&1&0\\0&0&1&0&0\\0&1&0&0&0\end{array}\right)
$$
together with the transformation $z\mapsto1/z$ in the arguments.
\end{ex}

Transformations of labels $0\leftrightarrow 1$ and $1\leftrightarrow z$ generate the transformation group $S_3$ of the three external vertices $0$, $1$, and $z$. So, every transformation of external vertices can
be expressed as a sequence of the transformations $0\leftrightarrow 1$ and $1\leftrightarrow z$.

\subsection{Product factorization}\label{sectproduct}
If the graph $G$ of a three-point function or a graphical function disconnects upon removal of the three external vertices, $G=G_1\cup G_2$ with $G_1\cap G_2\subseteq\{0,1,z\}$,
then the Feynman integral trivially factorizes into Feynman integrals over the internal vertices of $G_1$ and $G_2$. This implies
\begin{equation}
f_G^\alpha(z)=f_{G_1}^{\beta_1}(z)f_{G_2}^{\beta_2}(z),
\end{equation}
where, after the elimination of contractions, $\alpha=(\beta_1\cup\beta_2)\setminus (\beta_1\cap\beta_2)$.

\section{Appending an edge}\label{sectappendedge}
\begin{figure}
\begin{center}
\includegraphics{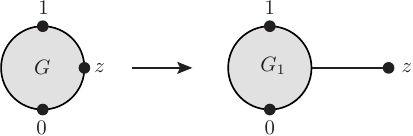}
\end{center}
\caption{Appending an edge to the vertex $z$ in $G$ gives $G_1$.}
\label{fig:append}
\end{figure}

In this section, we present the main calculation technique for graphical functions: appending an edge to a graph $G$ at the external vertex $z$, thus creating a graph $G_1$ with a new vertex $z$, see Figure \ref{fig:append}.
The appended edge is scalar with weight 1. Appending a sequence of edges and external differentiation at
$z=z_2$ allows one also to append edges with non-zero spin and weights $\neq1$, see Section \ref{sectappgen}.

Because the appended edge has spin 0, weight 1, we obtain from (\ref{Qdelta}) for $k=1$,
\begin{equation}
\Delta_{z_2}A^\alpha_{G_1}(z_0,z_1,z_2)=-\frac4{\Gamma(\lambda)}A^\alpha_G(z_0,z_1,z_2),
\end{equation}
where $\Delta_{z_2}=\partial^{\alpha_1}_{z_2}\partial^{\alpha_1}_{z_2}$ is the $D$-dimensional Laplace operator.

In the following section, we will translate this differential equation into an effective Laplace equation
\begin{equation}\label{eqF1fromf}
\mybox_\lambda^\alpha f^\alpha_{G_1}(z)=-\frac1{\Gamma(\lambda)}f^\alpha_G(z)
\end{equation}
acting on the graphical function $f^\alpha_{G_1}(z)$.

We use the symbol $\mybox$ for differential operators that act on the variable $z\in\CC$.
The subscript $\lambda$ refers to the dimension $D$ while the superscript $\alpha$ carries the spin structure.

We will then show how to construct the unique solution of the effective Laplace equation in the space of graphical functions.

\subsection{The effective Laplace operator $\mybox_\lambda^\alpha$}
We first determine the effect of the differential operators $\partial_{z_i}^\beta$ on a
graphical function $f_G^\alpha(z)$.
We consider $f_G^\alpha(z)$ as a function of the invariants $z\zz$ and $(z-1)(\zz-1)$, see (\ref{eqinvs}).
Let $\partial_s$ be the differential with respect to the invariant $(z-s)(\zz-s)$ for $s=0,1$.

We define the differential operators
\begin{equation}
\delta_k=\frac1{z-\zz}(z^k\partial_z-\zz^k\partial_\zz),\qquad k=0,1,2.
\end{equation}
Note that by explicit calculation
\begin{equation}\label{delta2}
\delta_2=-z\zz\delta_0+(z+\zz)\delta_1.
\end{equation}
Moreover, $\delta_0$, $\delta_1$, $\delta_2$ span a solvable three-dimensional Lie-Algebra, $[\delta_0,\delta_1]=0$, $[\delta_0,\delta_2]=\delta_0$, $[\delta_1,\delta_2]=\delta_1$.

For every component $f_{G,\pi}(z)$ of $f_G^\alpha(z)$ ($\pi\in\Pi_2^\alpha$, see (\ref{fGcomps})) we obtain
\begin{equation}\label{diffz_zz}
\partial_zf_{G,\pi}(z)=\zz\partial_0f_{G,\pi}(z)+(\zz-1)\partial_1f_{G,\pi}(z),\qquad\partial_\zz f_{G,\pi}(z)=z\partial_0f_{G,\pi}(z)+(z-1)\partial_1f_{G,\pi}(z).
\end{equation}
This yields
\begin{equation}\label{diff01}
\partial_0=\delta_1-\delta_0,\qquad\partial_1=-\delta_1.
\end{equation}
Using (\ref{fdef}) we obtain
\begin{equation}\label{pz2A}
\partial_{z_2}^\beta A_G^\alpha(0,z_1,z_2)=\|z_1\|^{-2\lambda N_G}\partial_{z_2}^\beta f_G^\alpha(z)\,=\,\|z_1\|^{-2\lambda N_G}\sum_{\pi\in\Pi_2^\alpha}\partial_{z_2}^\beta f_{G,\pi}(z)\,g^{\pi^0}\hat z_1^{\pi^1}\hat z_2^{\pi^2}.
\end{equation}
The variable $z$ is connected to $z_2$ by the invariants (\ref{eqinvs}). This yields in analogy to (\ref{diffz_zz}),
\begin{equation}
\partial_{z_2}^\beta f_{G,\pi}(z)\,\hat z_2^{\pi^2}=
\frac{2z_2^\beta}{\|z_1\|^2}\partial_0f_{G,\pi}(z)\,\hat z_2^{\pi^2}+\frac{2(z_2^\beta-z_1^\beta)}{\|z_1\|^2}\partial_1f_{G,\pi}(z)\,\hat z_2^{\pi^2}+f_{G,\pi}(z)\sum_{\gamma\in\pi^2}\frac{g^{\beta,\gamma}}{\|z_1\|}\,\hat z_2^{\pi^2_1}\cdots\widehat{\hat z_2^\gamma}\cdots\hat z_2^{\pi^2_n},
\end{equation}
where $n=|\pi^2|$. We consider the vectors $\hat z_2^\gamma$ as independent
variables and define
\begin{equation}
\partial_{\hat z_2^\gamma}\hat z_2^\beta=\delta_{\beta,\gamma}=g^{\beta,\gamma}.
\end{equation}
From Equation (\ref{diff01}) we obtain
\begin{equation}
\partial_{z_2}^\beta f_{G,\pi}(z)\,g^{\pi^0}\hat z_1^{\pi^1}\hat z_2^{\pi^2}=\|z_1\|^{-1}\Big(2\hat z_2^\beta(\delta_1-\delta_0)
+2(\hat z_2^\beta-\hat z_1^\beta)(-\delta_1)+\partial_{\hat z_2^\beta}\Big)f_{G,\pi}(z)\,g^{\pi^0}\hat z_1^{\pi^1}\hat z_2^{\pi^2}.
\end{equation}
Upon differentiation with respect to $z_2^\beta$ the weight of the graph $G$ increases
by $1/2\lambda$, so that we determine from (\ref{pz2A}) the effect of $\partial_{z_2}^\beta$ on the
graphical function $f_G^\alpha(z)$ as
\begin{equation}\label{pz2}
\partial_{z_2}^\beta\rightarrow 2\hat z_1^\beta\delta_1-2\hat z_2^\beta\delta_0+\partial_{\hat z_2^\beta}.
\end{equation}
A similar but slightly more complicated calculation using (\ref{delta2}) yields
for $\partial_{z_1}^\beta$ the correspondence
\begin{equation}\label{pz1}
\partial_{z_1}^\beta\rightarrow -\hat z_1^\beta(2\lambda N_G+|\pi^1|+|\pi^2|+2\delta_2)+2\hat z_2^\beta\delta_1+\partial_{\hat z_1^\beta},
\end{equation}
where we write
\begin{equation}
|\pi^i|=\sum_{\gamma\in\alpha}\hat z_i^\gamma\partial_{\hat z_i^\gamma}
\end{equation}
for the size of the $i$-component in the $\pi\in\Pi^\alpha_2$, $i=1,2$.

Finally, because the Feynman integral $A_G^\alpha(z_0,z_1,z_2)$ is a function of $z_1-z_0$ and $z_2-z_0$,
we have
\begin{equation}\label{pz0}
\partial_{z_0}^\beta=-\partial_{z_1}^\beta-\partial_{z_2}^\beta.
\end{equation}
For the calculation of the Laplace operator $\Delta_{z_2}=\partial_{z_2}^\beta\partial_{z_2}^\beta$, we observe that the first two terms in (\ref{pz2}) square
to the scalar effective $D$-dimensional Laplace operator $4\mybox_\lambda$, where \cite{gf} (with a proof in non-integer dimensions in \cite{parG})
\begin{equation}\label{DeltaD}
\mybox_\lambda=\partial_z\partial_\zz-\frac{\lambda}{z-\zz}(\partial_z-\partial_\zz)=\partial_z\partial_\zz-\lambda\delta_0.
\end{equation}
The mixed terms are
\begin{equation}
4(\hat z_1^\beta\delta_1-\hat z_2^\beta\delta_0)\partial_{\hat z_2^\beta}=4(\delta_1\hat z_1^\beta\partial_{\hat z_2^\beta}-|\pi_2|\delta_0).
\end{equation}
Altogether, we obtain the correspondence
\begin{equation}\label{boxalpha}
\frac{\Delta_{z_2}}4\rightarrow\mybox_\lambda^\alpha=\mybox_{\lambda+|\pi_2|}+\delta_1\hat z_1^\beta\partial_{\hat z_2^\beta}+\frac14\partial_{\hat z_2^\beta}\partial_{\hat z_2^\beta}.
\end{equation}
If one sorts the vector $f_{G,\pi}(z)$ by the number of $\hat z_2$ factors, then the
matrix form of $\mybox_\lambda^\alpha$ is triangular with the scalar effective Laplace operator \(\mybox_{\lambda+|\pi_2|}\) of dimension $D+2|\pi_2|$ on the diagonal.

The inversion of $\mybox_\lambda^\alpha$ can be reduced to the inversion of the scalar effective Laplace operators in $D+2j$ dimensions, $j=0,\ldots,|\alpha|$.
This problem is solved for even integer dimensions in \cite{gfe}. An extension to non-integer dimensions (in dimensional regularization) is in \cite{7loops}.

\begin{ex}
For spin $1$ we have $f_G^{\alpha_1}(z)=f_1(z)\hat z_1^{\alpha_1}+f_2(z)\hat z_2^{\alpha_1}$. The matrix of $\mybox_\lambda^{\alpha_1}$ is
\begin{equation}\label{matrix1}
\left(\begin{array}{cc}\mybox_\lambda&\delta_1\\0&\mybox_{\lambda+1}\end{array}\right).
\end{equation}
\end{ex}

\begin{ex}
For spin $2$, see Example \ref{expermext}, we represent $f_G^{\alpha_1,\alpha_2}(z)$ as the five-tuple $f_0(z),\ldots,f_4(z)$ of scalar graphical functions according to (\ref{eqpermext}).
The matrix of the spin $2$ effective Laplace operator $\mybox_\lambda^{\alpha_1,\alpha_2}$ is
\begin{equation}\label{matrix2}
\left(\begin{array}{ccccc}\mybox_\lambda&&&&\frac12\\
&\mybox_\lambda&\delta_1&\delta_1&\\
&&\mybox_{\lambda+1}&&\delta_1\\
&&&\mybox_{\lambda+1}&\delta_1\\
&&&&\mybox_{\lambda+2}\end{array}\right),
\end{equation}
where empty entries are zero.
\end{ex}

\subsection{Inverting $\mybox_\lambda^\alpha$ in the regular case}\label{sectappint}
The $D$-dimensional effective Laplace operator $\mybox_\lambda^\alpha$ can be represented by a triangular matrix whose diagonal is populated by scalar $D+2j$ dimensional
effective Laplace operators $\mybox_{\lambda+j}$ for $j=0,1,\ldots,|\alpha|$. To append an edge to a graph $G$, these $\mybox_{\lambda+j}$ need to be inverted.

Here we consider the situation that in dimension $D=2n+4-\epsilon$, $n=0,1,2,\ldots$, the limit $\epsilon=0$ is convergent.
In this case (which contains convergent graphical functions in integer dimensions), we call the graphical function regular.

In the regular case, the inversion of $\mybox_{\lambda+j}$ is unique in the space of scalar $D+2j$-dimensional graphical functions \cite{gfe,7loops}. There exists an efficient algorithm for inverting $\mybox_{\lambda+j}$ in the function space 
of generalized single-valued hyperlogarithms (GSVHs) \cite{GSVH}.
For low loop orders (typically $\leq7$), the space of GSVHs is sufficiently
general to perform one-scale QFT calculations. At higher loop orders, it is known that
GSVHs will not suffice (see e.g.\ \cite{K3}).

In the following, we will extend the inversion
of $\mybox_\lambda$ to positive spin by constructing an algorithm for the
inversion of $\mybox_\lambda^\alpha$.
We will see that a subtlety arises from poles at $z=1$.

Formally, the effective Laplace operator $\mybox_\lambda^\alpha$ can be inverted by
splitting it into the diagonal part $\mybox_{\lambda+|\pi_2|}$ and the nilpotent
part $\delta_1\hat z_1^\beta\partial_{\hat z_2^\beta}+\frac14\partial_{\hat z_2^\beta}\partial_{\hat z_2^\beta}$, see (\ref{boxalpha}).
By expanding the geometric series we obtain
\begin{equation}
(\mybox_\lambda^\alpha)^{-1}=\sum_{k=0}^{|\alpha|}\Big(-(\mybox_{\lambda+|\pi_2|})^{-1}\Big(\delta_1\hat z_1^\beta\partial_{\hat z_2^\beta}+\frac14\partial_{\hat z_2^\beta}\partial_{\hat z_2^\beta}\Big)\Big)^k(\mybox_{\lambda+|\pi_2|})^{-1}.
\end{equation}
Alternatively, one can use a bootstrap algorithm that constructs the inverse from
more $\hat z_2$ factors to less $\hat z_2$ factors (bottom up in (\ref{matrix1}) and (\ref{matrix2})). Concretely, we recursively solve the effective Laplace equation
\begin{equation}\label{fF}
(\mybox_\lambda^\alpha)^{-1}f^\alpha(z)=F^\alpha(z)
\end{equation}
by extracting the term $f_k^\alpha(z)$ of $f^\alpha(z)$ with the maximum $k$ number of factors of $\hat z_2$ in the
vector decomposition. In the first step of the algorithm this typically corresponds to the component $(\emptyset,\emptyset,\alpha)\in\Pi_2^\alpha$ with $k=|\alpha|$. We have
\begin{equation}\label{eqfz2}
f^\alpha(z)=f_k^\alpha(z)+\;\text{terms with $<k$ factors of $\hat z_2$.}
\end{equation}
The corresponding term $F_k^\alpha(z)$ in $F^\alpha(z)$ is given by the inversion
of $\mybox_{\lambda+k}$ (bottom right corners in (\ref{matrix1}) and (\ref{matrix2})),
\begin{equation}\label{eqFz2}
F_k^\alpha(z)=\mybox_{\lambda+k}^{-1}f_k^\alpha(z).
\end{equation}
From (\ref{boxalpha}) we obtain
\begin{align}\label{eqFf1}
F^\alpha(z)&=F_k^\alpha(z)+(\mybox_\lambda^\alpha)^{-1}g_k^\alpha(z),\quad\text{where}\\\nonumber
g_k^\alpha(z)&=f^\alpha(z)-f_k^\alpha(z)-
\delta_1\hat z_1^\beta\partial_{\hat z_2^\beta}F_k^\alpha(z)-\frac14\partial_{\hat z_2^\beta}\partial_{\hat z_2^\beta}F_k^\alpha(z).
\end{align}
The function $g_k^\alpha(z)$ has $\leq k-1$ factors of $\hat z_2$.
We continue solving (\ref{fF}) with $f^\alpha(z)\rightarrow g_k^\alpha(z)$
until we reach the scalar case $\alpha=\emptyset$ with $g_0(z)=0$.
Finally, we obtain $F^\alpha(z)=\sum_{k=0}^{|\alpha|}F_k^\alpha(z)$, see Example \ref{exlong}.

The advantage of the bootstrap algorithm is that each $\mybox_{\lambda+j}$, 
$j=0,\ldots,|\alpha|$, has to be inverted only once.

\begin{ex}\label{exexplicit}
For $|\alpha|=1$ we write $f_G^{\alpha_1}(z)=f_1(z)\hat z_1^{\alpha_1}+f_2(z)\hat z_2^{\alpha_1}$, see Example \ref{exextedge}. We obtain
\begin{equation}\label{eqF1}
F^{\alpha_1}(z)=\mybox_\lambda^{-1}\big(f_1(z)-\delta_1\mybox_{\lambda+1}^{-1}f_2(z)\big)\;\hat z_1^{\alpha_1}+\mybox_{\lambda+1}^{-1}f_2(z)\;\hat z_2^{\alpha_1}.
\end{equation}

For $|\alpha|=2$ we use the notation of Example \ref{expermext}. We obtain
\begin{align}\label{eqF2}
F^{\alpha_1,\alpha_2}(z)&=\mybox_\lambda^{-1}\big(f_0(z)-\frac12\mybox_{\lambda+2}^{-1}f_4(z)\big)\;g^{\alpha_1,\alpha_2}\\\nonumber
&+\mybox_\lambda^{-1}\big(f_1(z)-\delta_1\mybox_{\lambda+1}^{-1}(f_2(z)+f_3(z))+2\delta_1\mybox_{\lambda+1}^{-1}\delta_1\mybox_{\lambda+2}^{-1}f_4(z)\big)\;
\hat z_1^{\alpha_1}\hat z_1^{\alpha_2}\\\nonumber
&+\mybox_{\lambda+1}^{-1}\big(f_2(z)-\delta_1\mybox_{\lambda+2}^{-1}f_4(z)\big)\;
\hat z_1^{\alpha_1}\hat z_2^{\alpha_2}\\\nonumber
&+\mybox_{\lambda+1}^{-1}\big(f_3(z)-\delta_1\mybox_{\lambda+2}^{-1}f_4(z)\big)\;
\hat z_1^{\alpha_2}\hat z_2^{\alpha_1}\\\nonumber
&+\mybox_{\lambda+2}^{-1}f_4(z)\;\hat z_2^{\alpha_1}\hat z_2^{\alpha_2}.
\end{align}
\end{ex}
The main difficulty is to identify the right functions in the pre-image of
$\mybox_{\lambda+j}$ (i.e.\ to control the kernel of $\mybox_{\lambda+j}$).
In the scalar case this is facilitated by an analysis of the singular
structure of the pre-images. Theorem 36 in \cite{gfe} ensures that the pre-image is unique in the space of graphical functions.
When we extend this approach to positive spin, a naive inversion
of $\mybox_{\lambda+j}$ will not suffice.

We use the general structure of scalar graphical functions which are proved to have singularities only at $z=0,1$, or $\infty$ \cite{par}.
At $z=s$, $s=0,1$, the coefficients of scalar graphical functions have single-valued log-Laurent expansions \cite{gfe}:
\begin{equation}\label{01expansion}
f_{G,\pi}(z)=\sum_{\ell=0}^\VGint\sum_{m,\mm=M_s}^\infty c_{\ell,m,\mm}^{\pi,s}[\log(z-s)(\zz-s)]^\ell(z-s)^m(\zz-s)^\mm\quad\text{if }|z-s|<1,
\end{equation}
for some constants $c_{\ell,m,\mm}^{\pi,s}\in\RR$ and $M_s\in\ZZ$.
At infinity, there exist coefficients $c_{\ell,m,\mm}^{\pi,\infty}\in\RR$ and $M_\infty\in\ZZ$ such that
\begin{equation}\label{inftyexpansion}
f_G(z)=\sum_{\ell=0}^\VGint\sum_{m,\mm=-\infty}^{M_\infty}c_{\ell,m,\mm}^{\pi,\infty}(\log z\zz)^\ell z^m\zz^\mm\quad\text{if }|z|>1.
\end{equation}

Including the spin structure, the poles at $s=0,1$ are sums of
\begin{equation}\label{term}
[\log((z-s)(\zz-s))]^\ell(z-s)^m(\zz-s)^\mm(\hat z_2-\hat z_s)^{\beta_1}\cdots(\hat z_2-\hat z_s)^{\beta_j},
\end{equation}
with $\hat z_0=0$ and $\beta_i\in\alpha$, $i=1,\ldots,j$.
If $j<|\alpha|$, the term (\ref{term}) is multiplied by factors of $g$ and $\hat z_1$ to form a spin $|\alpha|$ object.

At $z\to s$, $\hat z_2\to\hat z_s$ these terms scale like $\log^\ell(|z-s|^2)|z-s|^{m+\mm+j}$.
In $D$ dimensions, integration over poles is regular if
\begin{equation}
m+\mm+j>-D.
\end{equation}
Note that spin $j>0$ relaxes the condition $m+\mm>-D$ for the regularity of scalar graphical functions.
So, in general, the scalar coefficients of regular graphical functions with spin
can have higher total pole orders $-m-\mm$ than the coefficients of scalar graphical functions.
If this is the case, we cannot naively adopt the algorithm for inverting the scalar effective Laplace
operator. The situation is alleviated if we have to invert $\mybox_{\lambda+j}$
which is the scalar effective Laplace operator in $D+2j$ dimensions (allowing for higher divergences),
but we will see that the problem is only postponed to later steps of calculating $F^\alpha_k(z)$.

At $z=\infty$, the pole order can only increase by including the spin structure
(factors $\hat z_2$). Hence, the coefficients are not more singular than in the scalar
case and no extra attention is necessary.
\begin{ex}\label{expole}
The function
\begin{equation}\label{eqfgen}
\frac{(\hat z_2^{\alpha_1}-\hat z_1^{\alpha_1})\cdots(\hat z_2^{\alpha_k}-\hat z_1^{\alpha_k})}{((z-1)(\zz-1))^n}
\end{equation}
is regular at $z=1$ in $D$ dimensions if $2n-k<D$. In particular,
\begin{equation}\label{eqf1}
\frac{\hat z_2^{\alpha_1}-\hat z_1^{\alpha_1}}{((z-1)(\zz-1))^2}
\end{equation}
is regular in four dimensions although $1/((z-1)(\zz-1))^2$ is singular at $z=1$.
Likewise
\begin{equation}\label{eqf2}
\frac{(\hat z_2^{\alpha_1}-\hat z_1^{\alpha_1})(\hat z_2^{\alpha_2}-\hat z_1^{\alpha_2})(\hat z_2^{\alpha_3}-\hat z_1^{\alpha_3})}{((z-1)(\zz-1))^4}
\end{equation}
is regular for $D=6$ although $1/((z-1)(\zz-1))^4$ has a (non-logarithmic) singularity at $z=1$ in six dimensions.
\end{ex}
For $s=0$ we have $\hat z_s=0$ and the term (\ref{term}) is in an entry of the vector
graphical function which has $j$ factors of $\hat z_2$.
We need to invert $\mybox_{\lambda+j}$, corresponding to dimension $D+2j$, in this sector. Because $m+\mm>-(D+j)\geq-(D+2j)$ the inversion is unique in the space
of scalar graphical functions. We do not need any adjustments at $s=0$.

For $s=1$ the situation is more complex: The term (\ref{term}) populates (with alternating signs) a selection of entries in the vector graphical function.
One entry, e.g., is the coefficient of $\hat z_1^{\beta_1}\cdots\hat z_1^{\beta_j}$ on which $\Delta_\lambda$ needs to be inverted. If $-D\geq m+\mm>-D-j$, the inversion is not unique in the space of scalar graphical functions and therefore ambiguous.

The ambiguity comes from the kernel of the scalar Laplace operator $\mybox_\lambda$ in integer dimensions ($\lambda=1,2,\ldots$).
With the notation of \cite{gfe} (Equations (9) and (10)) we have
\begin{equation}
\mybox_\lambda=\frac1{(z-\zz)^\lambda}\Delta_{\lambda-1}(z-\zz)^\lambda,\quad\Delta_{\lambda-1}=\partial_z\partial_\zz+\frac{\lambda(\lambda-1)}{(z-\zz)^2}.
\end{equation}
(For consistency with \cite{gfe} we use the notation $\Delta_{\lambda-1}$ which must not be confused
with the $D$-dimensional Laplace operator $\Delta_{z_2}$ that was used above.)

The kernel of $\mybox_\lambda$ is hence $(z-\zz)^{-\lambda}$ times the kernel of $\Delta_{\lambda-1}$. The latter was determined in Theorem 33 of \cite{gfe}.
We conclude that
\begin{equation}
\ker\mybox_\lambda=(z-\zz)^{-\lambda}(d_\lambda g(z)+\overline{d_\lambda} h(\zz)),
\end{equation}
where $g(z)$ and $h(\zz)$ are arbitrary holomorphic and anti-holomorphic functions,
respectively. The differential operator $d_\lambda$ is given by ($\overline{d_\lambda}$ is the complex conjugate of $d_\lambda$)
\begin{equation}
d_\lambda=\frac1{(z-\zz)^{\lambda-1}}\sum_{k=0}^{\lambda-1}(-1)^k\frac{(\lambda+k-1)!}
{(\lambda-k-1)!k!}(z-\zz)^{\lambda-k-1}\partial_z^{\lambda-k-1}.
\end{equation}

We find that $(z-\zz)^{-\lambda}d_\lambda$ is homogeneous in $z$ and $\zz$ of degree $-2\lambda+1$.
Poles $((z-1)(\zz-1))^{-k}$ come from functions $h(z)=(z-1)^{-2k+2\lambda-1}$
(and $h(\zz)=-\overline{g(z)}$ by symmetry). For a non-zero result we need $-2k+2\lambda-1<0$
(otherwise no pole is generated by $d_\lambda$). We conclude that poles
in $\ker\mybox_\lambda$ are of order $\geq 2\lambda$. The smallest example is
the function $((z-s)(\zz-s))^{-\lambda}\in\ker\mybox_\lambda$ for all $s\in\CC$.

It is proved in Theorem 5 of \cite{gfe} that the maximum pole orders
in $0$ and $1$ of a scalar graphical function $f_{G_1}(z)$ in Figure \ref{fig:append} is less than $2\lambda$.
(The stronger statement that the pole order is $\leq2\lambda-2$ uses the fact that
a scalar graphical function has even pole order which is not true for a
graphical function with spin, see Example \ref{expole}.)
Because the proof only uses scaling arguments (which are identical in the presence of spin) it carries over
to $f_{G_1}^\alpha(z)$. We thus search
for a regular function $F_{\mathrm{reg}}^\alpha(z)$ in the pre-image of $\mybox_\lambda^\alpha$ which inherits the constraints from the singularity structure
of the graphical function $f_{G_1}^\alpha(z)$.

Assume we generate a term (\ref{term}) in the kernel of $\mybox_\lambda^\alpha$.
The expression (\ref{term}) has a component with $j$ factors $\hat z_2$.
The coefficient of this part must be in the kernel of $\mybox_{\lambda+j}$.
This implies that $-m-\mm\geq2\lambda+2j$. It follows that the pole order
$-m-\mm-j$ of (\ref{term}) is $\geq2\lambda+j\geq2\lambda$.

Because the graphical function has pole order strictly less than $2\lambda$
we can kill the kernel which arises from singularities at $z=1$ by subtracting
all poles of order $\geq2\lambda$.

It is necessary to regularize functions by subtracting poles in $z=1$ at each step before the inversion of $\mybox_{\lambda+j}$ is applied in order to calculate $F^\alpha_j$. This way, the inversion
is well-defined in terms of the algorithm that appends an edge to a scalar graphical function in $D+2j$ dimensions.
The result will behave well on the singularities at $0$ and $\infty$.
The subtraction in the individual steps change the result only by pole terms which are
canceled in the end when all poles of order $\geq 2\lambda$ are subtracted.
Therefore, we need no compensation for these individual subtractions.
This subtraction algorithm generalizes without changes to the regular case in non-integer dimensions.
See Example \ref{exlong} for a detailed illustration.

\begin{ex}\label{short}
An explicit calculation gives for $s=0,1$,
\begin{align}\label{boxd1}
\mybox_\lambda((z-s)(\zz-s))^n&=n(n+\lambda)((z-s)(\zz-s))^{n-1},\\\nonumber
\delta_1((z-s)(\zz-s))^n&=-ns((z-s)(\zz-s))^{n-1}.
\end{align}
We consider the function (\ref{eqf1}) in four dimensions ($\lambda=1$).
We use (\ref{eqF1}) for $f_2(z)=-f_1(z)=((z-1)(\zz-1))^{-2}$ and obtain from (\ref{boxd1})
\begin{align*}
\mybox_2^{-1}f_2(z)&=-((z-1)(\zz-1))^{-1}\\
\delta_1\mybox_2^{-1}f_2(z)&=-((z-1)(\zz-1))^{-2}.
\end{align*}
Hence, in (\ref{eqF1}) we have $f_1(z)-\delta_1\mybox_2^{-1}f_2(z)=0$ which has the unique inverse $0$ (with respect to $\mybox_1$) in the space of graphical functions.
We obtain
$$
F^{\alpha_1}(z)=-\frac{\hat z_2^{\alpha_1}}{(z-1)(\zz-1)}
$$
which has a pole of order $2$ at $z=1$. We expand $F^{\alpha_1}(z)$ at $z=1$ yielding
$$
F^{\alpha_1}(z)=-\frac{\hat z_2^{\alpha_1}-\hat z_1^{\alpha_1}}{(z-1)(\zz-1)}-\frac{\hat z_1^{\alpha_1}}{(z-1)(\zz-1)}.
$$
The first term has pole order $1$ while the second term is a pole term of order $2=2\lambda$ (which we subtract). The regular solution is
$$
F^{\alpha_1}_{\mathrm{reg}}(z)=-\frac{\hat z_2^{\alpha_1}-\hat z_1^{\alpha_1}}{(z-1)(\zz-1)}.
$$
\end{ex}

\begin{ex}\label{exlong}
Consider the spin $3$ graphical function in (\ref{eqf2}) in six dimensions ($\lambda=2$).
We have (see Equation (\ref{eqfz2}))
$$
f_3^\alpha(z)=\frac{\hat z_2^{\alpha_1}\hat z_2^{\alpha_2}\hat z_2^{\alpha_3}}{((z-1)(\zz-1))^4}
$$
and obtain (see Equations (\ref{eqFz2}) and (\ref{boxd1}))
$$
F_3^\alpha(z)=-\frac{\hat z_2^{\alpha_1}\hat z_2^{\alpha_2}\hat z_2^{\alpha_3}}{6((z-1)(\zz-1))^3}.
$$
From Equation (\ref{eqFf1}) we get
\begin{align*}
g_3^\alpha(z)&=\frac{-\hat z_1^{\alpha_1}\hat z_1^{\alpha_2}\hat z_1^{\alpha_3}
+\hat z_1^{\alpha_1}\hat z_1^{\alpha_2}\hat z_2^{\alpha_3}+\hat z_1^{\alpha_1}\hat z_2^{\alpha_2}\hat z_1^{\alpha_3}+\hat z_2^{\alpha_1}\hat z_1^{\alpha_2}\hat z_1^{\alpha_3}-\frac12(\hat z_1^{\alpha_1}\hat z_2^{\alpha_2}\hat z_2^{\alpha_3}+\hat z_2^{\alpha_1}\hat z_1^{\alpha_2}\hat z_2^{\alpha_3}+\hat z_2^{\alpha_1}\hat z_2^{\alpha_2}\hat z_1^{\alpha_3})
}{((z-1)(\zz-1))^4}\\
&+\frac{g^{\alpha_1,\alpha_2}\hat z_2^{\alpha_3}+g^{\alpha_1,\alpha_3}\hat z_2^{\alpha_2}+g^{\alpha_2,\alpha_3}\hat z_2^{\alpha_1}}{12((z-1)(\zz-1))^3}.
\end{align*}
We read off
$$
f_2^\alpha(z)=-\frac{\hat z_1^{\alpha_1}\hat z_2^{\alpha_2}\hat z_2^{\alpha_3}+\hat z_2^{\alpha_1}\hat z_1^{\alpha_2}\hat z_2^{\alpha_3}+\hat z_2^{\alpha_1}\hat z_2^{\alpha_2}\hat z_1^{\alpha_3}}{2((z-1)(\zz-1))^4}.
$$
yielding
$$
F_2^\alpha(z)=\frac{\hat z_1^{\alpha_1}\hat z_2^{\alpha_2}\hat z_2^{\alpha_3}+\hat z_2^{\alpha_1}\hat z_1^{\alpha_2}\hat z_2^{\alpha_3}+\hat z_2^{\alpha_1}\hat z_2^{\alpha_2}\hat z_1^{\alpha_3}}{6((z-1)(\zz-1))^3}.
$$
With $F_2^\alpha$ we calculate $g_2^\alpha(z)$:
\begin{align*}
g_2^\alpha(z)&=-\frac{\hat z_1^{\alpha_1}\hat z_1^{\alpha_2}\hat z_1^{\alpha_3}}{((z-1)(\zz-1))^4}
+\frac{g^{\alpha_1,\alpha_2}(\hat z_2^{\alpha_3}-\hat z_1^{\alpha_3})+g^{\alpha_1,\alpha_3}(\hat z_2^{\alpha_2}-\hat z_1^{\alpha_2})+g^{\alpha_2,\alpha_3}(\hat z_2^{\alpha_1}-\hat z_1^{\alpha_1})}{12((z-1)(\zz-1))^3}.
\end{align*}
Now,
$$
f_1^\alpha(z)=\frac{g^{\alpha_1,\alpha_2}\hat z_2^{\alpha_3}+g^{\alpha_1,\alpha_3}\hat z_2^{\alpha_2}+g^{\alpha_2,\alpha_3}\hat z_2^{\alpha_1}}{12((z-1)(\zz-1))^3}
$$
and
$$
F_1^\alpha(z)=-\frac{g^{\alpha_1,\alpha_2}\hat z_2^{\alpha_3}+g^{\alpha_1,\alpha_3}\hat z_2^{\alpha_2}+g^{\alpha_2,\alpha_3}\hat z_2^{\alpha_1}}{24((z-1)(\zz-1))^2}
$$
leading to
\begin{align*}
g_1^\alpha(z)&=-\frac{\hat z_1^{\alpha_1}\hat z_1^{\alpha_2}\hat z_1^{\alpha_3}}{((z-1)(\zz-1))^4}.
\end{align*}
This is a pure pole term in $6$ dimensions. So, after subtraction,
$$
f_0^\alpha(z)=F_0^\alpha(z)=g_0^\alpha(z)=0.
$$
We expand $F^\alpha(z)=F_3^\alpha(z)+F_2^\alpha(z)+F_1^\alpha(z)+F_0^\alpha(z)$
in $z=1$, $\hat z_2=\hat z_1$, yielding
\begin{align*}
F^\alpha&(z)=-\frac{(\hat z_2^{\alpha_1}\!-\!\hat z_1^{\alpha_1})(\hat z_2^{\alpha_2}\!-\!\hat z_1^{\alpha_2})(\hat z_2^{\alpha_3}\!-\!\hat z_1^{\alpha_3})}{6((z-1)(\zz-1))^3}-
\frac{g^{\alpha_1,\alpha_2}(\hat z_2^{\alpha_3}\!-\!\hat z_1^{\alpha_3})+g^{\alpha_1,\alpha_3}(\hat z_2^{\alpha_2}\!-\!\hat z_1^{\alpha_2})+g^{\alpha_2,\alpha_3}(\hat z_2^{\alpha_1}\!-\!\hat z_1^{\alpha_1})}{24((z-1)(\zz-1))^2}\\
&+\frac{\hat z_1^{\alpha_1}\hat z_1^{\alpha_2}(\hat z_2^{\alpha_3}\!-\!\hat z_1^{\alpha_3})
+\hat z_1^{\alpha_1}\hat z_1^{\alpha_3}(\hat z_2^{\alpha_2}\!-\!\hat z_1^{\alpha_2})+
\hat z_1^{\alpha_2}\hat z_1^{\alpha_3}(\hat z_2^{\alpha_1}\!-\!\hat z_1^{\alpha_1})
}{6((z-1)(\zz-1))^3}+\frac{g^{\alpha_1,\alpha_2}\hat z_1^{\alpha_3}+g^{\alpha_1,\alpha_3}\hat z_1^{\alpha_2}+g^{\alpha_2,\alpha_3}\hat z_1^{\alpha_1}}{24((z-1)(\zz-1))^2}.
\end{align*}
The first two terms have pole order $3<2\lambda=4$.
They are consistent with the theory of graphical functions. Terms three and four
are pure poles in $z=1$ of orders $5$ and $4$, respectively.
They are killed by subtraction. The regular solution is given by the first two terms,
\begin{equation}\label{eqresultlong}
F_{\mathrm{reg}}^\alpha(z)=-\frac{(\hat z_2^{\alpha_1}\!-\!\hat z_1^{\alpha_1})(\hat z_2^{\alpha_2}\!-\!\hat z_1^{\alpha_2})(\hat z_2^{\alpha_3}\!-\!\hat z_1^{\alpha_3})}{6((z-1)(\zz-1))^3}-
\frac{g^{\alpha_1,\alpha_2}(\hat z_2^{\alpha_3}\!-\!\hat z_1^{\alpha_3})+g^{\alpha_1,\alpha_3}(\hat z_2^{\alpha_2}\!-\!\hat z_1^{\alpha_2})+g^{\alpha_2,\alpha_3}(\hat z_2^{\alpha_1}\!-\!\hat z_1^{\alpha_1})}{24((z-1)(\zz-1))^2}.
\end{equation}
It is possible to derive a result for the general function (\ref{eqfgen}) in any dimension $D$.
\end{ex}

\subsection{The log-divergent case}
If, after appending an edge, the Feynman integral has a logarithmic divergence, we cannot proceed as in Section \ref{sectappint}. One needs to subtract
the divergence taking into account the exact scaling behavior; an expansion in $\epsilon$
is not sufficient. From (\ref{convolution}) we, e.g., obtain for a pole of order $D+2a\epsilon$ in a scalar
propagator,
\begin{equation}
\int_{\RR^D}Q_{1+\frac{1+a\epsilon}\lambda}(x,y)Q_1(y,z)\frac{\dd^D y}{\pi^{D/2}}=-\frac{Q_{1+\frac{a\epsilon}\lambda}(x,z)}{a\epsilon(\lambda+a\epsilon)\Gamma(\lambda)}.
\end{equation}
An expansion in $\epsilon$ of the singular propagator gives
\begin{equation}
Q_{1+\frac{1+a\epsilon}\lambda}(x,y)=\frac1{\|x-y\|^D}(1-2a\epsilon\log\|x-y\|)+O(\epsilon^2).
\end{equation}
The coefficient of the pole in $\epsilon$ depends on the constant $a$ which determines the scaling behavior of the singularity. In the expansion of the propagator the parameter $a$ appears at order $\epsilon^1$
which may mix with other terms $\sim\|x-y\|^{-D}\log\|x-y\|$.

\begin{figure}
\centering
\includegraphics{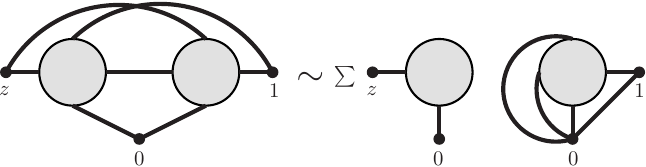}
\caption{The asymptotic expansion of graphical functions at $z=0$. Bold lines stand for sets of edges.}
\label{fig:asympt}
\end{figure}
The solution to this problem is to extract the exact scaling behavior of the singularity in terms of two-point functions as depicted in Figure \ref{fig:asympt} for a singularity at $z=0$.
The asymptotic expansions at $0$, $1$, and $\infty$ only rely on scaling arguments. They lift to graphical functions with numerator structure.
The result is related to the momentum space method of expansion by regions
(see \cite{Exbyreg} and the references therein).

The mathematical statement (a well-tested conjecture) is as follows \cite{numfunct,7loops}.
Assume $G$ is a graph such that the graphical function $f_G^\alpha(z)$ exists. Let $\sV^{\mathrm{int}}_G$ and $\sV^{\mathrm{ext}}_G=\{0,1,z\}$ be the sets of internal and external vertices of $G$,
respectively.  For $V\subseteq\sV^{\mathrm{int}}_G$ let $G[V]$ be the subgraph of $G$ which is induced by $V$, i.e.\ the subgraph which contains the vertices $V$ and all edges of $G$ with both
vertices in $V$. Further, let $G[V=0]$ be the contracted graph $G/G[V]$ where all vertices in $V$ are identified with the vertex $0$. Then, we obtain the asymptotic expansions at $z=0$ by
\begin{equation}\label{0asympt}
f_G^\alpha(z)=\sum_{V\subseteq\sV^{\mathrm{int}}_G}f_{G[V\cup\{0,z\}]}^\alpha(z)f_{G[V\cup\{0,z\}=0]}^\alpha(1+O(|z|))
\end{equation}
whenever the right-hand side exists. Note that the right-hand side of (\ref{0asympt}) only has two-point graphs.

The situation at $z=1$ is analogous,
\begin{equation}\label{1asympt}
f_G^\alpha(z)=\sum_{V\subseteq\sV^{\mathrm{int}}}f_{G[V\cup\{1,z\}]}^\alpha(z)f_{G[V\cup\{1,z\}=1]}^\alpha(1+O(|z-1|)).
\end{equation}

In case of a logarithmic singularity at infinity one has to subtract the asymptotic expansion
\begin{equation}\label{inftyasympt}
f_G^\alpha(z)=\sum_{V\subseteq\sV^{\mathrm{int}}}f_{G[V\cup\{0,1\}]}^\alpha f_{G[V\cup\{0,1\}=0]}^\alpha(z)(1+O(|z|^{-1})).
\end{equation}

Subtraction of the asymptotic formulae eliminate logarithmic singularities. Appending the edge to the subtractions (which are two-point functions) is a convolution that can be calculated analytically
using (\ref{convolution}). Details are analogous to the scalar case which is explained in detail in \cite{7loops}.

Note that in renormalizable QFTs, higher (non-logarithmic) divergences only appear in two-point insertions which can be eliminated with the method of Section \ref{secttwopt}. Thereafter, a physical graph
only has logarithmic singularities.

\subsection{Appending an edge with spin or with weight $\neq1$}\label{sectappgen}
Let $D=2\lambda+2=4+2n-\epsilon$, $n=0,1,\ldots$.

In this section we append an edge $Q^\alpha_\nu$ with spin $\alpha$ to the vertex $z$ of a graph $G$, thus creating
a new graph $G_1$ whose vertex $z$ connects to the vertex $z$ in $G$, see Figure \ref{fig:append}. This is only possible (with the method of this section)
for certain weights $\nu$. The range of all possible weights
is $\nu=1-k/\lambda+|\alpha|/2\lambda$ for $k=0,1,\ldots,n+1$.

The case $\nu=1$, $\alpha=\emptyset$ appends a scalar edge of weight $1$. The spin is only in the graphical function $f_G^\alpha(z)$. In this setup, appending an edge is equivalent to solving the differential
equation (\ref{eqF1fromf}). We have
\begin{equation}
f_{G_1}^\alpha(z)=-\frac1{\Gamma(\lambda)}(\mybox^\alpha_\lambda)^{-1}f_G^\alpha(z).
\end{equation}
An algorithm to determine $(\mybox^\alpha_\lambda)^{-1}f_G^\alpha(z)$ is explained in the previous subsections.
The result is unique in the space of graphical functions.

Next, we consider the scalar case of weight $\nu=1-k/\lambda$, $k=1,\ldots,n+1$.
We express the appended edge as a string of $k+1$ edges of weight $1$, see \cite{gfe}. Repeated use of
(\ref{convolution}) gives (see the proof of Proposition 39 in \cite{gfe})
\begin{align}
\label{eqappendk}
\begin{tikzpicture}[baseline={([yshift=-2.5ex]current bounding box.center)}]
    \coordinate (v1) at (0,0);
    \coordinate (v2) at (1.5,0);
    \filldraw[black!20] (v1) -- ([shift=(150:.2)]v1) arc (150:210:.2) -- (v1);
    \draw (v1) -- ([shift=(150:.2)]v1);
    \draw (v1) -- ([shift=(210:.2)]v1);
    \filldraw[black!20] (v2) -- ([shift=(30:.2)]v2) arc (30:-30:.2) -- (v2);
    \draw (v2) -- ([shift=(30:.2)]v2);
    \draw (v2) -- ([shift=(-30:.2)]v2);
    \filldraw (v1) circle (1.3pt);
    \filldraw (v2) circle (1.3pt);
    \draw (v1) -- node[above] {$1-\frac{k}{\lambda}$} (v2);
\end{tikzpicture}
\quad
&=
\quad
\frac{\displaystyle \Gamma(\lambda)^{k+1}k!}{\displaystyle\Gamma(\lambda-k)}
\quad
\underbrace{
\begin{tikzpicture}[baseline={([yshift=-2.5ex]current bounding box.center)}]
    \coordinate (v1) at (0,0);
    \coordinate (v2) at (1.5,0);
    \coordinate (v3) at (3,0);
    \coordinate (v3a) at (3.5,0);
    \node (v4) at (4.5,0) {$\cdots$};
    \coordinate (v5a) at (5.5,0);
    \coordinate (v5) at (6,0);
    \coordinate (v6) at (7.5,0);
    \filldraw[black!20] (v1) -- ([shift=(150:.2)]v1) arc (150:210:.2) -- (v1);
    \draw (v1) -- ([shift=(150:.2)]v1);
    \draw (v1) -- ([shift=(210:.2)]v1);
    \filldraw[black!20] (v6) -- ([shift=(30:.2)]v6) arc (30:-30:.2) -- (v6);
    \draw (v6) -- ([shift=(30:.2)]v6);
    \draw (v6) -- ([shift=(-30:.2)]v6);
    \filldraw (v1) circle (1.3pt);
    \filldraw (v2) circle (1.3pt);
    \filldraw (v3) circle (1.3pt);
    \filldraw (v5) circle (1.3pt);
    \filldraw (v6) circle (1.3pt);
    \draw[white] (v1) -- node[above] {$\phantom{\frac{k}{\lambda}}$} (v6);
    \draw (v1) -- node[above] {$1$} (v2) -- node[above] {$1$} (v3) -- (v3a);
    \draw (v5a) -- (v5) -- node[above] {$1$} (v6);
\end{tikzpicture}
}_{k + 1}.
\end{align}

The string of $k+1$ edges of weight $1$
can be appended by repeatedly inverting $\mybox^\alpha_\lambda$. It is not easily possible to append more than $n+2$ edges because one runs into non-logarithmic infrared singularities
for $k\geq n+2$.

If the appended propagator has spin $\alpha\neq\emptyset$, we append a scalar edge of
weight $\nu-|\alpha|/2\lambda$ (if possible) and differentiate the appended edge with
respect to $z_2$ using (\ref{pz2}).
Note that this procedure can produce extra singularities in intermediate steps so that even
a convergent graphical function may need to be dimensionally regularized.

\subsection{A simple test and first benchmarks}
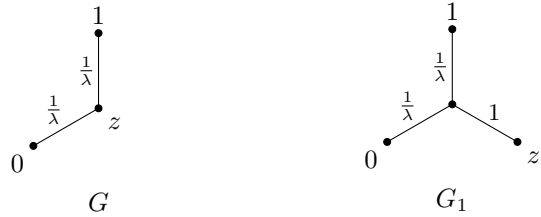
\begin{figure}
\begin{align*}
& 
\begin{tikzpicture}
    \coordinate[label=below right:$z$] (v) ;
    \def \rad {1}
    \coordinate[label=above:$1$] (v1) at ([shift=(90:\rad)]v);
    \coordinate[label=below left:$0$] (v2) at ([shift=(210:\rad)]v);
%
    \draw (v) -- node[inner sep=1pt,left] {\small{$\frac1\lambda$}} (v1);
    \draw (v) -- node[inner sep=1pt,above left] {\small{$\frac1\lambda$}} (v2);
%
    \filldraw (v) circle (1.3pt);
    \filldraw (v1) circle (1.3pt);
    \filldraw (v2) circle (1.3pt);
%
    \node [below=of v] {$G$};
\end{tikzpicture}
\hspace{3cm}
\begin{tikzpicture}
    \coordinate (v) ;
    \def \rad {1}
    \coordinate[label=above:$1$] (v1) at ([shift=(90:\rad)]v);
    \coordinate[label=below left:$0$] (v2) at ([shift=(210:\rad)]v);
    \coordinate[label=below right:$z$] (v3) at ([shift=(330:\rad)]v);
    \draw (v) -- node[inner sep=1pt,left] {\small{$\frac1\lambda$}} (v1);
    \draw (v) -- node[inner sep=1pt,above left] {\small{$\frac1\lambda$}} (v2);
    \draw (v) -- node[inner sep=1pt,above right] {\small{$1$}} (v3);
    \filldraw (v) circle (1.3pt);
    \filldraw (v1) circle (1.3pt);
    \filldraw (v2) circle (1.3pt);
    \filldraw (v3) circle (1.3pt);
    \node [below=of v] {$G_1$};
\end{tikzpicture}
\end{align*}
\caption{
Constructing the three-star (right) in $D=2\lambda+2$ dimensions by appending an edge to the two-star (left). The weights are as indicated.
}
\label{fig:3star}
\end{figure}

To test the method of appending an edge we applied it to the graph in Figure \ref{fig:3star}.
The two-star is rational
\begin{equation}
f_G(z)=\frac1{z\zz(z-1)(\zz-1)}.
\end{equation}
The three-star is easily calculated with \cite{Shlog} by appending an edge to the scalar graph
$G$. In four dimensions it contains a Bloch-Wigner dilogarithm \cite{gf}.
We want to obtain the graphical functions with spin $k_0+k_1$ by taking
$k_0$ derivatives with respect to $z_0$ and $k_1$ derivatives with respect to $z_1$ using
Equations (\ref{pz2}), (\ref{pz1}), and (\ref{pz0}).
Because each differentiation increases the pole order by one,
the graphical function is regular if $k_0,k_1\leq2\lambda-1$.

We do this in two different ways. Firstly, we take derivatives of the three-star itself.
Secondly, we take derivatives of the two-star and append an edge to the vertex $z$.
Both methods have to give the same result. This is checked for all configurations and
orders of $\epsilon$. We stopped the calculations if it takes one day on a single
core of an office PC. The memory demand in these cases is modest ($\approx$1GB).
We did this for three different orders in $\epsilon$ where we reached
dimension $10$ and spin $7$ for order $\epsilon^0$,
dimension $10$ and spin $5$ for order $\epsilon^2$,
and dimension $8$ and spin $4$ for order $\epsilon^4$.

\section{Integration over $z$}\label{intz}
\begin{figure}
\centering
\includegraphics{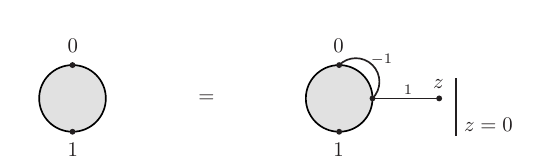}
\caption{Integration over $z$ by appending an edge of weight $1$.}
\label{fig:int}
\end{figure}
There exist three options for the transition from graphical functions to two-point functions and to scalar Feynman integrals.

Firstly, one can specify the external vertex $z$ to $0$, $1$, or $\infty$,
which transforms a three-point amplitude into a two-point amplitude. This simple method was used to calculate the zigzag integrals in \cite{ZZ}.

Secondly, one can integrate over $z$ in integer dimensions using a residue theorem which was developed in \cite{gf}. In non-integer dimensions,
the residue theorem cannot yet be used. A generalization to non-integer dimensions needs to control functions with powers of $\log|z-\zz|$ which are not GSVHs because they
are singular on the entire real axis. A method that handles this difficulty has not yet been developed.

Thirdly, one can integrate over $z$ by the following practical method.
One can add a scalar edge of weight $-1$ between the external vertices $0$ and $z$. Then, one appends a scalar edge of weight $1$ to $z$.
Finally, setting $z=0$ gives the integral of the original graph over $z$, see Figure \ref{fig:int}. This method is fully general but more time-consuming than the residue
theorem in integer dimension. It seems inefficient to calculate a 
graphical function in the intermediate step only to specialize it to $z=0$.
For scalar graphs, however, integration over $z$ was typically not the bottleneck of the calculations.
The situation is somewhat worse with non-zero spin (depending an the size of $|\alpha|$), so that ultimately
it may be useful to generalize the second method that uses the residue theorem.

\section{Constructible graphs}\label{construct}
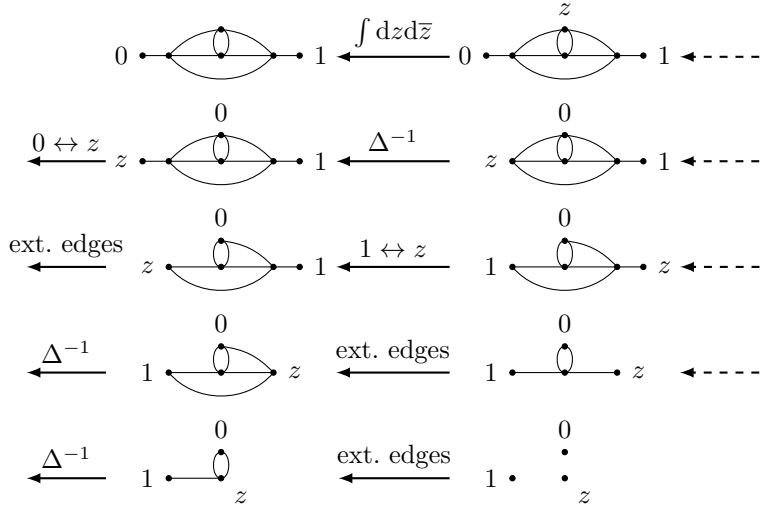
\begin{figure}
\centering
\begin{tikzpicture}[x=2.3ex,y=2.3ex]

    \begin{scope}[local bounding box=fullcateye]
        \coordinate (vm);
        \coordinate [above=1 of vm] (vo);
        \coordinate [left =2 of vm] (vl);
        \coordinate [right=2 of vm] (vr);
        \coordinate [left =1 of vl] (vll);
        \coordinate [right=1 of vr] (vrr);

        \draw (vm) to [bend left =90] (vo);
        \draw (vm) to [bend right=90] (vo);
        \draw (vl) to [bend left =20] (vo);
        \draw (vr) to [bend right=20] (vo);
        \draw (vl) to [bend right=50] (vr);

        \draw (vl) -- (vm);
        \draw (vr) -- (vm);
        \draw (vll) -- (vl);
        \draw (vrr) -- (vr);

        \node [left =.2 of vll] {$0$};
        \node [right=.2 of vrr] {$1$};

        \filldraw (vm) circle (1pt);
        \filldraw (vo) circle (1pt);
        \filldraw (vl) circle (1pt);
        \filldraw (vr) circle (1pt);
        \filldraw (vll) circle (1pt);
        \filldraw (vrr) circle (1pt);
    \end{scope}

    \begin{scope}[xshift=130,local bounding box=fullcateyez]
        \coordinate (vm);
        \coordinate [above=1 of vm] (vo);
        \coordinate [left =2 of vm] (vl);
        \coordinate [right=2 of vm] (vr);
        \coordinate [left =1 of vl] (vll);
        \coordinate [right=1 of vr] (vrr);

        \draw (vm) to [bend left =90] (vo);
        \draw (vm) to [bend right=90] (vo);
        \draw (vl) to [bend left =20] (vo);
        \draw (vr) to [bend right=20] (vo);
        \draw (vl) to [bend right=50] (vr);

        \draw (vl) -- (vm);
        \draw (vr) -- (vm);
        \draw (vll) -- (vl);
        \draw (vrr) -- (vr);

        \node [left =.2 of vll] {$0$};
        \node [right=.2 of vrr] {$1$};
        \node [above=.2 of vo ] {$z$};

        \filldraw (vm) circle (1pt);
        \filldraw (vo) circle (1pt);
        \filldraw (vl) circle (1pt);
        \filldraw (vr) circle (1pt);
        \filldraw (vll) circle (1pt);
        \filldraw (vrr) circle (1pt);
    \end{scope}

    \draw[-latex,thick] (fullcateye-|fullcateyez.west) -- (fullcateye) node[midway,above]{$\int \dd z \dd \zz$};
    \draw[-latex,thick,dashed]  ([xshift=30]fullcateye-|fullcateyez.east) -- (fullcateye-|fullcateyez.east)  node[midway,above]{};

    \begin{scope}[yshift=-40,local bounding box=fullcateyez2]
        \coordinate (vm);
        \coordinate [above=1 of vm] (vo);
        \coordinate [below=1 of vm] (vu);
        \coordinate [left =2 of vm] (vl);
        \coordinate [right=2 of vm] (vr);
        \coordinate [left =1 of vl] (vll);
        \coordinate [right=1 of vr] (vrr);

        \draw (vm) to [bend left =90] (vo);
        \draw (vm) to [bend right=90] (vo);
        \draw (vl) to [bend left =20] (vo);
        \draw (vr) to [bend right=20] (vo);
        \draw (vl) to [bend right=50] (vr);

        \draw (vl) -- (vm);
        \draw (vr) -- (vm);
        \draw (vll) -- (vl);
        \draw (vrr) -- (vr);

        \node [left =.2 of vll] {$z$};
        \node [right=.2 of vrr] {$1$};
        \node [above=.2 of vo ] {$0$};

        \node [below=.2 of vu ] {\phantom{$0$}};

        \filldraw (vm) circle (1pt);
        \filldraw (vo) circle (1pt);
        \filldraw (vl) circle (1pt);
        \filldraw (vr) circle (1pt);
        \filldraw (vll) circle (1pt);
        \filldraw (vrr) circle (1pt);
    \end{scope}

    \draw[-latex,thick]  (fullcateyez2) -- ([xshift=-30]fullcateyez2-|fullcateyez2.west)  node[midway,above]{$0 \leftrightarrow z$};

    \begin{scope}[yshift=-40,xshift=130,local bounding box=cateyeredl]
        \coordinate (vm);
        \coordinate [above=1 of vm] (vo);
        \coordinate [below=1 of vm] (vu);
        \coordinate [left =2 of vm] (vl);
        \coordinate [right=2 of vm] (vr);
        \coordinate [left =1 of vl] (vll);
        \coordinate [right=1 of vr] (vrr);

        \draw (vm) to [bend left =90] (vo);
        \draw (vm) to [bend right=90] (vo);
        \draw (vl) to [bend left =20] (vo);
        \draw (vr) to [bend right=20] (vo);
        \draw (vl) to [bend right=50] (vr);

        \draw (vl) -- (vm);
        \draw (vr) -- (vm);
        \draw (vrr) -- (vr);

        \node [left =.2 of vl] {$z$};
        \node [right=.2 of vrr] {$1$};
        \node [above=.2 of vo ] {$0$};
        \node [below=.2 of vu ] {\phantom{$0$}};
        \node [right=.2 of vrr ] {\phantom{$0$}};
        \node [left=.2 of vll ] {\phantom{$0$}};

        \filldraw (vm) circle (1pt);
        \filldraw (vo) circle (1pt);
        \filldraw (vl) circle (1pt);
        \filldraw (vr) circle (1pt);
        \filldraw (vrr) circle (1pt);
    \end{scope}

    \draw[-latex,thick] (fullcateyez2-|cateyeredl.west) -- (fullcateyez2) node[midway,above]{$\Delta^{-1}$};
    \draw[-latex,thick,dashed]  ([xshift=30]fullcateyez2-|cateyeredl.east) -- (fullcateyez2-|cateyeredl.east)  node[midway,above]{};

    \begin{scope}[yshift=-80,xshift=0,local bounding box=cateyeredl2]
        \coordinate (vm);
        \coordinate [above=1 of vm] (vo);
        \coordinate [below=1 of vm] (vu);
        \coordinate [left =2 of vm] (vl);
        \coordinate [right=2 of vm] (vr);
        \coordinate [left =1 of vl] (vll);
        \coordinate [right=1 of vr] (vrr);

        \draw (vm) to [bend left =90] (vo);
        \draw (vm) to [bend right=90] (vo);
        \draw (vr) to [bend right=20] (vo);
        \draw (vl) to [bend right=50] (vr);

        \draw (vl) -- (vm);
        \draw (vr) -- (vm);
        \draw (vrr) -- (vr);

        \node [left =.2 of vl] {$z$};
        \node [right=.2 of vrr] {$1$};
        \node [above=.2 of vo ] {$0$};
        \node [below=.2 of vu ] {\phantom{$0$}};
        \node [right=.2 of vrr ] {\phantom{$0$}};
        \node [left=.2 of vll ] {\phantom{$0$}};

        \filldraw (vm) circle (1pt);
        \filldraw (vo) circle (1pt);
        \filldraw (vl) circle (1pt);
        \filldraw (vr) circle (1pt);
        \filldraw (vrr) circle (1pt);
    \end{scope}

    \draw[-latex,thick]  (cateyeredl2) -- ([xshift=-30]cateyeredl2-|cateyeredl2.west)  node[midway,above]{\text{ext.~edges}};

    \begin{scope}[yshift=-80,xshift=130,local bounding box=cateyeredl3]
        \coordinate (vm);
        \coordinate [above=1 of vm] (vo);
        \coordinate [below=1 of vm] (vu);
        \coordinate [left =2 of vm] (vl);
        \coordinate [right=2 of vm] (vr);
        \coordinate [left =1 of vl] (vll);
        \coordinate [right=1 of vr] (vrr);

        \draw (vm) to [bend left =90] (vo);
        \draw (vm) to [bend right=90] (vo);
        \draw (vr) to [bend right=20] (vo);
        \draw (vl) to [bend right=50] (vr);

        \draw (vl) -- (vm);
        \draw (vr) -- (vm);
        \draw (vrr) -- (vr);

        \node [left =.2 of vl] {$1$};
        \node [right=.2 of vrr] {$z$};
        \node [right=.2 of vrr ] {\phantom{$0$}};
        \node [left=.2 of vll ] {\phantom{$0$}};
        \node [above=.2 of vo ] {$0$};
        \node [below=.2 of vu ] {\phantom{$0$}};

        \filldraw (vm) circle (1pt);
        \filldraw (vo) circle (1pt);
        \filldraw (vl) circle (1pt);
        \filldraw (vr) circle (1pt);
        \filldraw (vrr) circle (1pt);
    \end{scope}

    \draw[-latex,thick]  (cateyeredl2-|cateyeredl3.west) -- (cateyeredl2)  node[midway,above]{$1 \leftrightarrow z$};

    \draw[-latex,thick,dashed]  ([xshift=30]cateyeredl3-|cateyeredl3.east) -- (cateyeredl3)  node[midway,above]{};

    \begin{scope}[yshift=-120,local bounding box=cateyeredlr]
        \coordinate (vm);
        \coordinate [above=1 of vm] (vo);
        \coordinate [below=1 of vm] (vu);
        \coordinate [left =2 of vm] (vl);
        \coordinate [right=2 of vm] (vr);
        \coordinate [left =1 of vl] (vll);
        \coordinate [right=1 of vr] (vrr);

        \draw (vm) to [bend left =90] (vo);
        \draw (vm) to [bend right=90] (vo);
        \draw (vr) to [bend right=20] (vo);
        \draw (vl) to [bend right=50] (vr);

        \draw (vl) -- (vm);
        \draw (vr) -- (vm);

        \node [left =.2 of vl] {$1$};
        \node [right=.2 of vr] {$z$};
        \node [above=.2 of vo ] {$0$};
        \node [below=.2 of vu ] {\phantom{$0$}};
        \node [right=.2 of vrr ] {\phantom{$0$}};
        \node [left=.2 of vll ] {\phantom{$0$}};

        \filldraw (vm) circle (1pt);
        \filldraw (vo) circle (1pt);
        \filldraw (vl) circle (1pt);
        \filldraw (vr) circle (1pt);
    \end{scope}

    \draw[-latex,thick]  (cateyeredlr) -- ([xshift=-30]cateyeredlr-|cateyeredlr.west)  node[midway,above]{$\Delta^{-1}$};

    \begin{scope}[yshift=-120,xshift=130,local bounding box=cateyeredlr2]
        \coordinate (vm);
        \coordinate [above=1 of vm] (vo);
        \coordinate [below=1 of vm] (vu);
        \coordinate [left =2 of vm] (vl);
        \coordinate [right=2 of vm] (vr);
        \coordinate [left =1 of vl] (vll);
        \coordinate [right=1 of vr] (vrr);

        \draw (vm) to [bend left =90] (vo);
        \draw (vm) to [bend right=90] (vo);

        \draw (vl) -- (vm);
        \draw (vr) -- (vm);

        \node [left =.2 of vl] {$1$};
        \node [right=.2 of vr] {$z$};
        \node [above=.2 of vo ] {$0$};
        \node [below=.2 of vu ] {\phantom{$0$}};
        \node [right=.2 of vrr ] {\phantom{$0$}};
        \node [left=.2 of vll ] {\phantom{$0$}};

        \filldraw (vm) circle (1pt);
        \filldraw (vo) circle (1pt);
        \filldraw (vl) circle (1pt);
        \filldraw (vr) circle (1pt);
    \end{scope}

    \draw[-latex,thick]  (cateyeredlr-|cateyeredlr2.west) -- (cateyeredlr)  node[midway,above]{\text{ext.~edges}};
    \draw[-latex,thick,dashed]  ([xshift=30]cateyeredlr2-|cateyeredlr2.east) -- (cateyeredlr2)  node[midway,above]{};

    \begin{scope}[yshift=-160,xshift=0,local bounding box=cateyeredlr3]
        \coordinate (vm);
        \coordinate [above=1 of vm] (vo);
        \coordinate [below=1 of vm] (vu);
        \coordinate [left =2 of vm] (vl);
        \coordinate [right=2 of vm] (vr);
        \coordinate [left =1 of vl] (vll);
        \coordinate [right=1 of vr] (vrr);

        \draw (vm) to [bend left =90] (vo);
        \draw (vm) to [bend right=90] (vo);

        \draw (vl) -- (vm);

        \node [left =.2 of vl] {$1$};
        \node [below right=.2 of vm] {$z$};
        \node [above=.2 of vo ] {$0$};
        \node [below=.2 of vu ] {\phantom{$0$}};
        \node [right=.2 of vrr ] {\phantom{$0$}};
        \node [left=.2 of vll ] {\phantom{$0$}};

        \filldraw (vm) circle (1pt);
        \filldraw (vo) circle (1pt);
        \filldraw (vl) circle (1pt);
    \end{scope}

    \draw[-latex,thick]  (cateyeredlr3) -- ([xshift=-30]cateyeredlr3-|cateyeredlr3.west)  node[midway,above]{$\Delta^{-1}$};

    \begin{scope}[yshift=-160,xshift=130,local bounding box=one]
        \coordinate (vm);
        \coordinate [above=1 of vm] (vo);
        \coordinate [below=1 of vm] (vu);
        \coordinate [left =2 of vm] (vl);
        \coordinate [right=2 of vm] (vr);
        \coordinate [left =1 of vl] (vll);
        \coordinate [right=1 of vr] (vrr);

        \node [left =.2 of vl] {$1$};
        \node [below right=.2 of vm] {$z$};
        \node [above=.2 of vo ] {$0$};
        \node [below=.2 of vu ] {\phantom{$0$}};
        \node [right=.2 of vrr ] {\phantom{$0$}};
        \node [left=.2 of vll ] {\phantom{$0$}};

        \filldraw (vm) circle (1pt);
        \filldraw (vo) circle (1pt);
        \filldraw (vl) circle (1pt);
    \end{scope}

    \draw[-latex,thick]  (one) -- ([xshift=-42]one-|one.west)  node[midway,above]{\text{ext.~edges}};
\end{tikzpicture}
\caption{Example of the graphical reduction of a two-point function in $\phi^4$ theory.
The two-point function can be calculated by following the arrows starting from the trivial graphical function which is equal to $1$. The symbol $\Delta^{-1}$ means appending an edge
and $\int\dd z\dd\zz$ is integration over $z$.}
\label{fig:GFreduction}
\end{figure}
We combine the results of the previous sections to define a class of graphs with spin whose graphical functions in even dimensions $\geq4$ can be computed
(subject to constraints from time and memory consumption).

Graphical functions depend on the position of the external labels $0$, $1$, $z$.
A permutation of external vertices changes the graphical function in an explicit way, see Section \ref{sectpermute}. To decide whether a graphical function is computable
or not, the distinction between internal and external labels suffices;
the assignment of $0$, $1$, $z$ to the external labels is insignificant.
We may hence forget the labels on external vertices in the following definition.

\begin{defn}[An extension of Section 3.7 in \cite{gf}]\label{defconstructible}
We consider the following set of commuting reduction steps for a graphical function $f_G^\alpha(z)$:
\begin{enumerate}
\item[(R1)] Deletion of edges between external vertices, Section \ref{sectextedge},
\item[(R2)] Product factorization, Section \ref{sectproduct},
\item[(R3)] Integrating out two-point functions, Section \ref{secttwopt},
\item[(R4)] Contraction of an isolated edge with spin $\alpha$ and weight $1-k/\lambda+|\alpha|/2\lambda$, $k=0,1,\ldots,n+1$, attached to any external vertex, Sections \ref{sectpermute} and \ref{sectappendedge}.
\end{enumerate}
A graphical function is {\em irreducible} if it cannot be reduced by any of these steps.
Any two-point insertion in (R3) is a two-point subgraph in $G$. Its Feynman integral equals a sum of propagators with coefficients that are scalar
Feynman integrals and metric tensors, see Section \ref{secttwopt}.
Maximum use of the reduction steps (R1) to (R4) maps a graph $G$ to a set of scalar Feynman integrals and irreducible graphical functions.
The {\em kernel} of $G$ is this unique set of scalar Feynman integrals and irreducible graphical functions.

We recursively define constructible graphical functions and scalar Feynman integrals by
\begin{enumerate}
\item A graphical function is constructible if, after using (R3), it is logarithmically divergent and its kernel is the empty graphical function plus constructible scalar Feynman integrals.
\item A scalar Feynman integral $P_G$ is constructible if it has two vertices or if there exist three vertices $a,b,c$ in $G$ such that the graphical function $G|_{abc=01z}$ is constructible.
\end{enumerate}
\end{defn}
See Figure \ref{fig:GFreduction} for the construction of the cat eye two-point graph.

In \cite{gfe} constructible graphical functions were defined for completed graphs (see \cite{gf}). Completion (adding a fourth external vertex $\infty$ using conformal
invariance) is only possible for theories whose fermionic interaction is of Yukawa type,
see Section \ref{sectYukawa}.
In such theories the vertices are scalar, i.e.\ they have no vector or spin index. In particular, QED and Yang-Mills theories are not amenable to completion.

\section{Primitive scalar Feynman integrals in Yukawa-$\phi^4$ theory}\label{sectYukawa}
In this section we apply the results of the article to four-dimensional Yukawa-$\phi^4$ theory which has a spin
0 boson with a four-point $\phi^4$ interaction (the Higgs in the standard model) and a spin $1/2$ fermion with a three-point Yukawa coupling to the boson, see Figure \ref{fig:yukawaphi4}.
\begin{figure}
\centering
\includegraphics{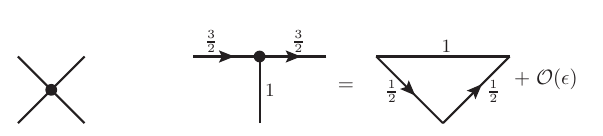}
\caption{The interactions in Yukawa-$\phi^4$ theory (the Gross-Neveu-Yukawa model). The three-point Yukawa interaction is rational in four dimensions (uniqueness)}
\label{fig:yukawaphi4}
\end{figure}
In physics, this Yukawa-$\phi^4$ theory has the name Gross-Neveu-Yukawa model.
The beta functions of the interactions have been calculated with classical momentum space IBP methods to loop order four \cite{GNY}.

As a consequence of conformal symmetry, the Yukawa vertex is rational in position space for $\epsilon=0$.
This results in the star-triangle ($Y-\Delta$) identity that is depicted in Figure \ref{fig:yukawaphi4}.
Such identities are called uniqueness relations in physics \cite{Kunique}. Note that the
right-hand side of the star-triangle identity is not a subgraph of Yukawa-$\phi^4$ theory. The star-triangle
identity operates on a more general set of Feynman graphs that also has spin $1/2$ particles of weight $1/2$.
In the following, we work in this set of generalized Feynman graphs.

The star-triangle identity preserves the total weights on the external vertices in four dimensions. Every vertex has weighted degree four.
This implies that the fermion propagator of weight $1+1/2\lambda=3/2+O(\epsilon)$ on the left-hand side
changes to weight
$1-1/2\lambda=1/2+O(\epsilon)$ while the boson weight changes from $1$ to $1/\lambda=1+O(\epsilon)$.
(The full formula, which includes all terms of $O(\epsilon)$, has three more graphs on the right-hand side.)
In position space, each application of the star-triangle identity eliminates one integration.
Hence, we want to use the star-triangle identity as often as possible. Because it is not clear
in which sequence and for which vertices we should use the star-triangle identity to get maximum
reduction, we apply the identity in every possible way. However, we do not use it from right to left
(as a triangle-star identity) which, in certain situations, may be necessary to obtain maximum reductions.
A similar situation exists in $\phi^3$ theory where the graph-theoretic structure of star-triangle reductions
was analyzed in \cite{JYDY}.

In this section, we only skim the results for primitive graphs (graphs without subdivergence).
The amplitudes of primitive graphs have a pole in $\epsilon$ of order one whose residue is given as a finite
integral in four dimensions. This primitive (scalar) Feynman integral (Feynman period or Feynman residue) is universal,
it does not depend on the renormalization scheme. The primitive Feynman integral contributes to the corresponding beta function.
The highest transcendental weight part of the beta function at a given loop order only comes from primitive graphs.

Primitive Feynman integrals can be calculated in exactly four dimensions ($\epsilon=0$).
In the presence of conformal symmetry this simplifies their computation. This compensates
somewhat for the intrinsic difficulty of primitive graphs. Still, the calculation of primitive Feynman integrals
is often the hardest step in completing a loop order. In $\phi^4$ theory and in six-dimensional $\phi^3$ theory,
e.g., it was possible to calculate the beta function to the highest loop orders in which all primitive Feynman integrals are known
($7$ loops and $6$ loops, respectively \cite{7loops,6lphi3}).

There exists a long history of calculating primitive Feynman integrals in $\phi^4$ theory \cite{BK,Census}.
Mathematically, primitive Feynman integrals are particularly interesting.
They are conjectured to be a comodule under the motivic Galois coaction \cite{coaction,Bcoact1,Bcoact2},
which led to the discovery of motivic Galois structures in QFTs (the cosmic Galois group).

The purpose of this section is to show that the concept is consistent and efficient in a non-trivial
physical QFT. We chose Yukawa-$\phi^4$ theory because, due to conformal symmetry, one has a large class of constructible completed graphs (see Section \ref{completion}). We hence get more results
in higher loop orders. The complete theory of graphical functions in primitive Yukawa-$\phi^4$
theory is quite intricate and will be presented in a separate article.

\subsection{Inversion}
We consider the following inversion on $x\in\RR^D$,
\begin{equation}
x\mapsto\tilde x=\frac{x}{x^2},
\end{equation}
which keeps the direction of $x$ but reverses its length, $||x||\mapsto1/||x||$.
For any weight $\nu\in\RR$ this leads to the identities
\begin{equation}\label{spin0inv}
\frac1{||\tilde x-\tilde y||^{2\lambda\nu}}\mapsto||x||^{2\lambda\nu}\frac1{||x-y||^{2\lambda\nu}}\,||y||^{2\lambda\nu}
\end{equation}
for spin 0 and
\begin{equation}\label{spin12inv}
\frac{\tilde x\!\!\!/-\tilde y\!\!\!/}{||\tilde x-\tilde y||^{2\lambda\nu+1}}\mapsto-\hat x\!\!\!/\,||x||^{2\lambda\nu}\frac{x\!\!\!/-y\!\!\!/}{||x-y||^{2\lambda\nu+1}}\,||y||^{2\lambda\nu}\hat y\!\!\!/
\end{equation}
for spin $1/2$. Here, $\hat x=x/||x||$ is a unit vector, so that $\hat x\!\!\!/^2=\hat y\!\!\!/^2=1$.
If the fermion vertex does not carry a $\gamma$-matrix (like in Yukawa theory but unlike QED), then the left and right factors $\hat x\!\!\!/$ and $\hat y\!\!\!/$
square to 1 in a fermion line and the fermion behaves (up to boundary terms) like a signed boson line. In such theories we can use inversion to
simplify the calculation of Feynman integrals.

\subsection{Completion}\label{completion}
\begin{figure}
\begin{center}
\includegraphics[scale=1.11]{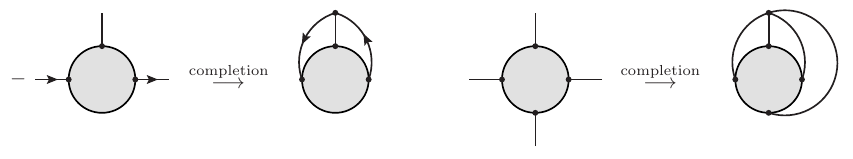}
\end{center}
\caption{The completions of primitive three-point and four-point graphs.}
\label{fig:Y4completion}
\end{figure}

Inversion was already used by D. Broadhurst and D. Kreimer in \cite{BK} to obtain identities between primitive Feynman integrals in $\phi^4$ theory.
This method was systematized in \cite{Census}. It was shown that all graphs that, by inversion, have the same Feynman integral form an equivalence class
which can graphically be represented by a single `completed' graph. In the completed graph the external legs of the vertex graphs are
connected to an extra vertex $\infty$. After completion we forget the label $\infty$ and all vertices are of the same type.
Completed graphs are $4$-regular in $\phi^4$ theory and $3$-regular in $\phi^3$ theory. Whatever vertex one chooses to open in a completed graph, one gets primitive graphs
with the same Feynman integral.

The benefit of completion is that the number of completed graphs is significantly smaller than the number of primitive graphs.
Moreover (and mostly), it solves the calculation of the Feynman integrals for all decompletions if it is possible to calculate one decompletion.

In Yukawa-$\phi^4$ theory, completed graphs have two types of vertices: a three-valent Yukawa vertex and a four-valent $\phi^4$ vertex. Decompletion at a Yukawa vertex gives a minus sign
(reminiscent of the minus sign in (\ref{spin12inv})), while decompletion at a $\phi^4$ vertex gives a plus sign, see Figure \ref{fig:Y4completion}.

The completion of a primitive three-point graph may be identical to the completion of a primitive four-point graph.
However, the contribution to the four-point beta function is one loop order lower than the contribution to the three-point beta function because the number of
loops in the completed graphs goes down by three upon decompleting at a $\phi^4$ vertex while it goes down by two upon decompleting at a Yukawa vertex.
So, the primitive contributions to the two beta functions (Yukawa and $\phi^4$) are connected over different loop orders via completion.

We denote the completion of a graph $G$ by $\overline{G}$,
\begin{equation}
G\quad\stackrel{\rm completion}{\longrightarrow}\quad\overline{G},
\end{equation}
and obtain for the amplitude of a primitive graph $G$
\begin{equation}
A_G(x)=\frac{P_{\overline{G}}}{\epsilon}+O(1),
\end{equation}
where $P_{\overline G}$ is the Feynman integral of any (and hence all) of the decompletions of $\overline G$.

Note that, although the completed graph $\overline{G}$ looks like vacuum graph, it is an equivalence class of
primitive graphs. The calculation of $P_{\overline{G}}$ always implies to choose a vertex $\infty$ first.
In fact, we choose three vertices $0,1,\infty$ and calculate the Feynman integral of $\overline{G}\setminus\infty$ with external vertices $0$ and $1=\hat z_1$, see Section \ref{sectper}.

To calculate $P_{\overline G}$ with graphical functions, we choose a fourth external vertex $z$ corresponding to the vector $\hat z_2=z_2/\|z_1\|$ and calculate the graphical function of the graph. Thereafter, we integrate over $z$, see Section \ref{construct}.
We use the powerful theory of completed graphical functions, see \cite{gf,gfe} in the scalar case.
In the presence of spin, the theory of completed graphical functions is significantly more intricate than the theory of completed scalar Feynman integrals.
In the theory of completed graphical functions it is important to keep fermions as spin $1/2$ particles
and not to evaluate the traces over $\gamma$ matrices in the first step. An early evaluation of the $\gamma$
traces transforms the Feynman integral into a large sum of Feynman integrals with spin $1$ propagators.
Because spin $1$ propagators do not behave well under inversion (in contrast to spin zero and spin $1/2$,
see (\ref{spin0inv}) and (\ref{spin12inv})), most of the terms in the sum are much more complicated than the original Feynman integral.

It can be proved that a completed graph $\overline{G}$ has finite Feynman integral $P_{\overline{G}}$ if
and only if any edge cut of total weight $\leq4$ separates off a vertex ($\overline{G}$ is internally weight
$5$ connected), see Figures \ref{fig:smallex}, \ref{fig:largeex}. Physically, this means that $\overline{G}$ has no non-trivial three-point or four-point insertions.

\subsection{Some identities}
\begin{figure}
\begin{center}
\includegraphics[scale=1]{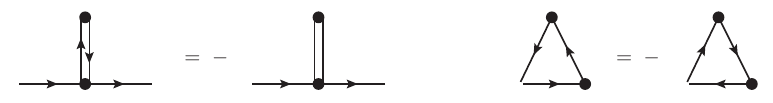}
\end{center}
\caption{Spin identities for a double line and for a triangle.}
\label{fig:Y4ids1}
\end{figure}
Yukawa-$\phi^4$ theory has a rich structure of identities. Beyond the star-triangle identity
in Figure \ref{fig:yukawaphi4} we have the identities in Figure \ref{fig:Y4ids1}. These identities
act solely on the spin. The weights of the propagators are unchanged. Choosing suitable coordinates,
they are graphical versions of the elementary equations
\begin{equation}
x\!\!\!/(-x\!\!\!/)=-x^2,\quad x\!\!\!/(y\!\!\!/-x\!\!\!/)(-y\!\!\!/)=-y\!\!\!/(x\!\!\!/-y\!\!\!/)(-x\!\!\!/).
\end{equation}
The last equation is true because both sides equal $x^2y\!\!\!/-y^2x\!\!\!/$. Note that the latter identity renders a Feynman integral zero if the graph is symmetric upon
interchanging the two fat vertices on the right hand side of Figure \ref{fig:Y4ids1}.

\begin{figure}
\begin{center}
\includegraphics[scale=0.98]{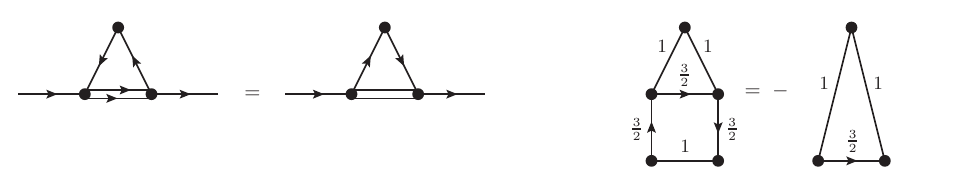}
\end{center}
\caption{Reductions obtained from the identities in Figure \ref{fig:Y4ids1} and from the star-triangle identity
in Figure \ref{fig:yukawaphi4}.}
\label{fig:Y4ids2}
\end{figure}
Together with the star-triangle identity, these two identities give rise to the identities in
Figure \ref{fig:Y4ids2}.
On the left-hand side, a triangle loop may be reduced to a triangle without fermion loop.
The edges may have any weights. The identity on the right-hand side is a reduction inside
Yukawa-$\phi^4$ theory. Both sides of the equation can be considered as subgraphs
of (completed) Yukawa-$\phi^4$ graphs.

\begin{figure}
\begin{center}
\includegraphics[scale=0.98]{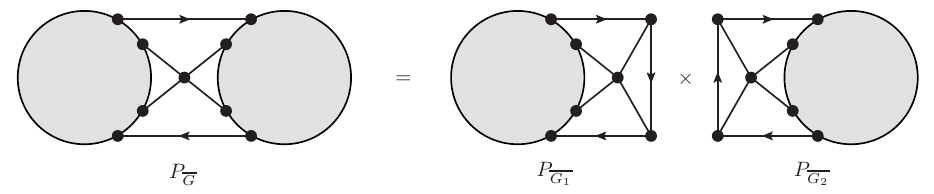}
\end{center}
\caption{The factor identity in Yukawa-$\phi^4$ theory.}
\label{fig:Y4factor}
\end{figure}
The Feynman integral of a completed primitive graph factorizes if it has a three vertex split
(label the split vertices $0,1,\infty$ and the Feynman integral factorizes).
The only (convergent) configuration in Yukawa-$\phi^4$ theory is depicted in Figure \ref{fig:Y4factor}
where all fermion edges have weight $3/2$ and all boson edges have weight $1$,
\begin{equation}
P_{\overline G}=P_{\overline{G_1}}\,P_{\overline{G_1}}.
\end{equation}
More factorizations are possible in the generalized set of Feynman graphs which include fermion
edges of weight $1/2$.

\subsection{Small graphs}
\begin{figure}
\centering
\includegraphics[scale=0.95]{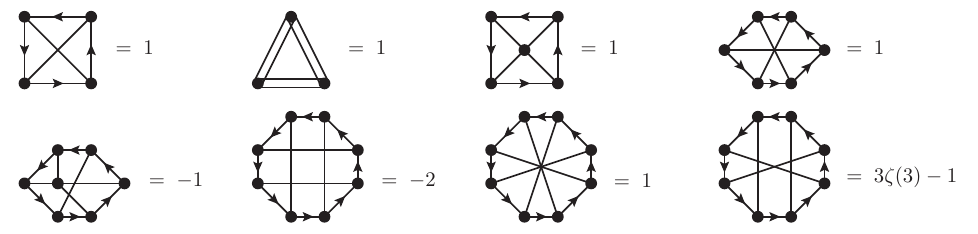}
\caption{Completed graphs in Yukawa-$\phi^4$ theory up to five loops. The Feynman integrals
contribute to the Yukawa beta function up to three loops and to the $\phi^4$ beta function up to two loops.}
\label{fig:smallex}
\end{figure}

In Figure \ref{fig:smallex} we list all completed graphs with $\leq 5$ loops in Yukawa-$\phi^4$ theory which
have a finite Feynman integral. Note that some of the graphs have $P_G=1$ by successive use of the star-triangle identity, see Figure \ref{fig:yukawaphi4}.

The only graph with non-rational Feynman integral contributes to the three-loop beta function of the Yukawa interaction.
\subsection{An eight loop example}\label{sect:eight}
\begin{figure}
\centering
\includegraphics[scale=0.6]{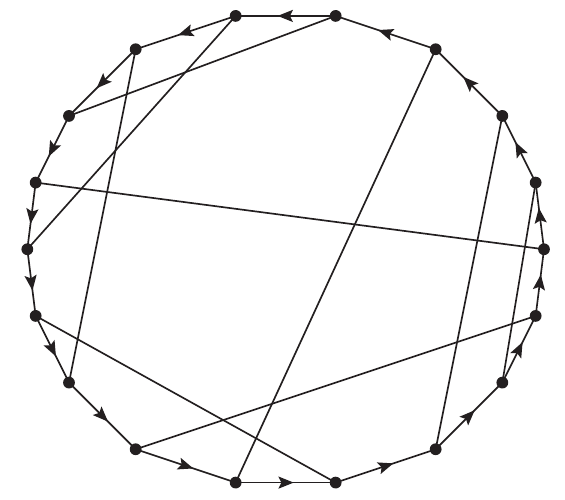}
\caption{A completed graph $\overline{G}$ whose Feanman integral gives a contribution to the Yukawa beta function at eight loops.}
\label{fig:largeex}
\end{figure}

Consider the completed graph $\overline{G}$ depicted in Figure \ref{fig:largeex}.
A computation of 15 minutes with {\tt HyperlogProcedures} \cite{Shlog} on a single core of an office PC consuming 2GB of RAM gives the result
\begin{align}\label{P8}
P_{\overline G}&=\frac{15}2\zeta(5)-39\zeta(3)^2+\frac{175}8\zeta(7)-143\zeta(3)\zeta(5)\\\nonumber
&-\,44\zeta(3)^3+\frac{16745}{36}\zeta(9)+332\zeta(3)\zeta(7)+165\zeta(5)^2\\\nonumber
&+\,\frac3{10}\pi^6\zeta(5)-\frac{189}{50}\pi^4\zeta(7)-\frac{1701}2\pi^2\zeta(9)-\frac{567}5\zeta(5,3,3)\\\nonumber
&-\,\frac52\zeta(3)^2\zeta(5)+\frac{160377}{20}\zeta(11)\\\nonumber
&-\,14\zeta(3)^4+\frac{18985}{72}\zeta(3)\zeta(9)-\frac{4485}8\zeta(5)\zeta(7)\\\nonumber
&=2.7071150367306270291466\ldots.
\end{align}
To the author's knowledge it is impossible to do this calculation (even numerically) with any other method. The number $P_{\overline G}$ is
contained in the $\QQ$-span of Feynman integrals in pure $\phi^4$ theory.

\subsection{The status of Yukawa-$\phi^4$ theory}
For a given completed graph $\overline{G}$ there exist many different ways (thousands if $\overline{G}$ has seven loops) to
choose the external vertices $0,1,z,\infty$ such that the Feynman integral can be calculated with
graphical functions. All results agree, which establishes a strong test for the correctness
of the implementation.

Up to four loops in the primitive graphs (six or seven loops in the completed graphs) all primitive Feynman integrals could
be calculated. By the time of writing, at five loops the evaluations of eleven out of 705 graphs are missing.
At higher loop orders the gaps grow rapidly. We calculated individual graphs up to eight loops in the
beta function of the Yukawa vertex, corresponding to 10 loops
in the completed graph, see Section \ref{sect:eight}.

The number content of the Feynman integrals that could be calculated is identical to $\phi^4$ theory.
In particular, the coaction conjectures also seem to hold in the presence of a Yukawa interaction
\cite{coaction}. However, until now the selection of graphs whose Feynman integral has been calculated
is biased towards simple graphs with not very rich number content.
Still, this result establishes another test for the validity of the method, because it is unlikely
that an incorrect calculation does not show a more generic number content than $\phi^4$ theory.

It is worth noting that Feynman integrals of graphs with fermion edges are more prone to weight mixing than Feynman integrals of graphs in pure $\phi^4$ theory, see, e.g.,
the last graph in Figure \ref{fig:smallex}.

In future work, one can use identities to relate unknown Feynman integrals to (sums of) known Feynman integrals.
Alternatively (ideally: additionally) one can calculate kernel graphical functions using identities
or by mapping them to scalar graphical functions. These can either be calculated by the existing
version of {\tt HyperlogProcedures} \cite{Shlog} or one can try to solve them with IBP methods.
A combined method that acts on Feynman integrals and graphical functions has been used successfully in six-dimensional $\phi^3$ theory \cite{recursive}.

We plan to compute the five loop contribution to the beta functions of Yukawa-$\phi^4$ theory in the near future.

\end{document}